\documentclass[10pt,reqno]{amsart}
\usepackage[T1]{fontenc}

\usepackage{times}
\usepackage{amsmath,amsthm,amssymb,latexsym}
\usepackage[letterpaper,margin=1in]{geometry}
\usepackage{etoolbox}
\patchcmd{\section}{\scshape}{\bfseries}{}{\PackageError{paper}{Section font patch failed}{Check the document class.}}
\patchcmd{\subsubsection}{\itshape}{\bfseries}{}{\PackageError{paper}{Subsubsection font patch failed}{Check the document class.}}
\patchcmd{\paragraph}{\normalfont}{{\normalfont\itshape}}{}{\PackageError{paper}{Paragraph font patch failed}{Check the document class.}}

\usepackage{graphicx,xcolor,enumerate,float,booktabs}
\usepackage[numbers,square,comma,sort&compress]{natbib}
\usepackage{xurl}
\usepackage[colorlinks=true,linkcolor=blue,citecolor=blue,urlcolor=blue]{hyperref}
\hypersetup{pdftitle={Explicit objective functions in modularity-based community detection on multiplex networks},pdfauthor={Elizaveta Evmenova and Petr Chunaev},pdfkeywords={modularity, community detection, multiplex networks}}
\usepackage[export]{adjustbox}
\usepackage{subcaption}
\usepackage{appendix}
\usepackage[section]{placeins}
\usepackage{multirow}

\definecolor{methodSF}{HTML}{1F77B4}
\definecolor{methodEF}{HTML}{2CA02C}
\definecolor{methodLF}{HTML}{D62728}
\newcommand{\toymethodlegend}{\par\smallskip{\small {\color{methodSF}$\bullet$} SF\quad {\color{methodEF}$\blacksquare$} EF\quad {\color{methodLF}$\blacktriangle$} LF}\par}
\newcommand{\partitionlegend}{%
  \par\smallskip\footnotesize
  {\color{methodSF}\rule[0.5ex]{1.5em}{0.6pt}} SF partitions\quad
  {\color{methodEF}\rule[0.5ex]{0.4em}{0.6pt}\hspace{0.15em}\rule[0.5ex]{0.4em}{0.6pt}\hspace{0.15em}\rule[0.5ex]{0.4em}{0.6pt}} EF partitions\quad
  {\color{methodLF}\rule[0.5ex]{0.65em}{0.6pt}\hspace{0.15em}\rule[0.5ex]{0.1em}{0.6pt}\hspace{0.15em}\rule[0.5ex]{0.65em}{0.6pt}} LF partitions\par}

\theoremstyle{plain}
\newtheorem{theorem}{Theorem}
\newtheorem{lemma}{Lemma}
\newtheorem{corollary}{Corollary}
\theoremstyle{remark}
\newtheorem{remark}{Remark}

\makeatletter
\AfterEndEnvironment{thebibliography}{\@setaddresses\global\let\@setaddresses\relax}
\makeatother
\title[Explicit objective functions in modularity-based community detection]{Explicit objective functions in modularity-based community detection on multiplex networks}
\author{Elizaveta Evmenova}
\address{Faculty of Electrical Engineering, Mathematics and Computer Science, Delft University of Technology, 2628 CD, Delft, Netherlands}
\email{e.evmenova.o@gmail.com}
\thanks{Elizaveta Evmenova: ORCID \href{https://orcid.org/0000-0002-6032-5830}{0000-0002-6032-5830}.}
\author{Petr Chunaev}
\address{KairosHelix, Inc., 4 Robinhood Road, Natick, MA 01760, USA}
\email{petr.chunaev@kairoshelix.ai}
\thanks{Petr Chunaev: corresponding author. ORCID \href{https://orcid.org/0000-0001-8169-8436}{0000-0001-8169-8436}.}
\date{}

\begin{document}
\begin{abstract}
Early, simultaneous, and late fusion combine network layers at different stages of community detection. Their modularity values consequently need not measure the same property of a partition. For normalized shared-node multiplexes, a unified comparison extends an existing two-layer identity: early fusion equals simultaneous fusion plus a weighted variance of community degree volumes. An elementary simplex argument explains simultaneous fusion's preference for single-layer weights. For a specified late-fusion construction, an additional transformation correction separates changes to the graph from changes to the objective. Exact small-network experiments and Louvain and Leiden experiments illustrate these distinctions. On label-free Cora, all objective-selected weights favor citations alone, while the fixed equal-weight mixtures achieve higher median adjusted Rand agreement with subject categories. The analysis provides a diagnostic for interpreting fusion objectives; higher objective values alone do not establish better recovery of externally defined communities.
\end{abstract}
\maketitle

\section{Introduction}
\label{sec:introduction}

Community detection (CD) in multiplex networks~\cite{Magnani2021} identifies groups of nodes based on interaction patterns across multiple network layers. Multiplex networks arise naturally in a wide range of systems, including social~\cite{steinhaeuser2008community, bedi2016community, Das2022}, economic~\cite{brauksa2013use, gui2014dynamic}, biological~\cite{novoa2021multi, buphamalai2021network}, and transportation systems~\cite{huang2018comparing}.

Recent surveys and taxonomies of multilayer community detection methods discuss different modeling assumptions and evaluation criteria~\cite{Magnani2021,Hamed2024,BoukabeneMetrics2025}. Building on this literature, we adopt the early-fusion (EF), simultaneous-fusion (SF), and late-fusion (LF) taxonomy introduced for node-attributed networks by~\cite{Chunaev2019survey} and extend it to multiplex networks under the standing assumptions stated in Section~\ref{sec:cd_formaulation_and_multiplex_graphs}:

(1) fuse layers first and perform CD on the resulting ordinary network (EF);

(2) solve the CD task using the information from all layers simultaneously (SF);

(3) apply CD to all layers separately and then fuse the results by finding a consensus among the resulting partitions (LF).

Many methods within these categories are specified algorithmically and evaluated empirically, making their underlying optimization criteria difficult to identify or compare directly \cite{Chunaev2019survey,Magnani2021,Vieira2020}. Empirical evaluations may also use external or task-specific measures of community quality \cite{Chunaev2019survey,ChunaevComplNetw2020,Chunaev2021}. While such measures are useful for assessing the resulting communities, they do not establish analytical relationships between the objectives optimized by different methods; more generally, as noted by Peel et al.'s No Free Lunch result, empirical superiority on a given collection of problems does not imply a general advantage across all community detection problems~\cite{Peel2017}. We therefore compare the objective functions analytically within the EF, SF, and LF taxonomy, and separately test their relationship to external label agreement on Cora without using those labels in graph construction or optimization.

We use a common notation to derive and compare the objectives of several modularity-based CD methods for multiplex networks. The same framework can also be applied to node-attributed networks, which can be represented as multiplex networks by treating network topology as one layer and node attributes as one or more attribute-similarity layers.

The contribution is a unified, operational comparison of these representative objectives, including the consequences of a specified LF transformation and a separate external validation. The EF--SF identity extends the previously published two-layer result~\cite{ChunaevComplNetw2020,Chunaev2021} through the standard weighted-variance identity; it is not a new general optimization principle. The main analytical conclusions are as follows. First, the objective of the {\sf EF} formulation equals the {\sf SF} objective plus a non-negative heterogeneity term that measures how much the total degree weight of a community varies between layers. Thus, the {\sf EF} objective assigns additional value to partitions in which a community has substantially different total degree weights across active layers; the heterogeneity term is zero precisely when these community-wise layer volumes are identical across all active layers. Second, we consider a layer-weight vector $\alpha$, whose elements reflect the importance assigned to each layer. The non-negative, normalized weights form a simplex. The elementary fact that an affine function on a simplex attains its maximum at a vertex gives the corresponding {\sf SF} statement. For the {\sf EF} objective, we show that, for any fixed partition, the objective is concave in $\alpha$. Nevertheless, when optimizing jointly over the partition and the layer weights, the optimum may occur with $\alpha$ in the interior of the simplex. Third, for the specific LF construction considered in this study, we show that the {\sf LF} objective can be decomposed into Newman modularity on the original layer graphs, a heterogeneity term, and a correction term induced by the graph transformation. For each layer evaluated at its own transient partition, this correction is non-negative and vanishes precisely when the transformation preserves the total within-community edge weight.

Our analysis considers representative modularity-based {\sf EF}, {\sf SF}, and one specific {\sf LF} construction for shared-node multiplex networks with normalized non-negative weights, no interlayer edges, and a common resolution parameter in Newman modularity. It does not address arbitrary multilayer networks, alternative LF consensus constructions, or generalized multilayer modularity with unconstrained node-layer assignments. We use brute-force optimization on synthetic multiplexes to verify and illustrate the analytical results, while experiments with Louvain and Leiden on real-world networks illustrate how heuristic optimization and preprocessing choices affect the observed results.

The label-free Cora experiment in Section~\ref{sec:cora_external} separates objective-based weight selection from agreement with subject categories, which are withheld from graph construction and optimization. All three fusion methods, with both Louvain and Leiden, select citations alone by their own objective. Yet the fixed equal-weight mixtures achieve higher median adjusted Rand index and normalized mutual information than either single-layer baseline. This result illustrates a difference between internal objective values and external agreement on one network; the subject categories are external annotations rather than definitive community ground truth.

The paper is organized as follows. Section~\ref{sec:problem_statement} defines the networks and the modularity-based CD problem. Section~\ref{sec:literature} discusses related work and the classification of methods using the adopted taxonomy. Section~\ref{sec:methods} describes the representative {\sf EF}, {\sf SF}, and {\sf LF} formulations and positions them within the existing literature. Section~\ref{sec:theorems} derives their objective functions, while Section~\ref{sec:synergy} studies the optimization of the layer weights $\alpha$. Section~\ref{sec:experiments} presents exact and heuristic experiments and the label-free external validation; Section~\ref{sec:discussion} discusses their interpretation and limitations. The longer proofs and supporting experimental detail are collected in Appendices~\ref{app:proofs}--\ref{app:cora_sensitivity}. Appendix~\ref{app:liu_lf} contrasts the chosen LF construction with a cluster-network alternative.

\section{Formulation of the Community Detection Problem on Multiplex Graphs}
\label{sec:cd_formaulation_and_multiplex_graphs}
\label{sec:problem_statement}
    This section defines the multiplex network setting and formulates the corresponding community detection problem. We first specify the multiplex representation used throughout the paper, including how node-attributed networks can be represented within this framework. We then define the community detection problem, the modularity objective used as the common building block for the methods considered later, and the associated layer-weight and normalization conventions.

\subsection{Definition of the Multiplex Network}
\label{sec:multiplex_networks}

    We consider a node-aligned \textit{multiplex network} $G$ consisting of $L$ weighted layers defined on a common node set $\mathcal{V} = \{ v_i \}_{i=1}^n$. Each layer $g_l=(\mathcal{V},\mathcal{E},W_l)$ is an ordinary undirected weighted graph, where $W_l=\{w_l(e_{ij})\}$ assigns a non-negative weight to each edge $e_{ij}$. Thus, a multiplex network can be represented as an indexed family of ordinary weighted graphs,
    \begin{equation*}
        \label{eq:multiplex-def}
        G=\{g_l\}_{l=1}^L,\qquad g_l=(\mathcal{V},\mathcal{E},W_l),\qquad l=1,\ldots,L,
    \end{equation*}
    where all layers share the same node set $\mathcal{V}$ and the same set $\mathcal{E}=\{e_{ij}=\{v_i,v_j\}:1\le i<j\le n\}$ of possible edges between distinct nodes. The input layers are loop-free; transient self-loops introduced by LF are handled separately below. For each layer $l$, $W_l = {w_l(e_{ij})}$ assigns a non-negative weight to each $e_{ij} \in \mathcal{E}$, with zero weight representing an absent edge.

    A \textit{node-attributed network}, in which nodes carry attributes in addition to the network topology, can also be represented as a multiplex network. In this representation, the network topology forms a single layer, while the node attributes are transformed into one or more attribute-similarity layers with non-negative edge weights, as, for example, in~\cite{Chunaev2019survey,ChunaevComplNetw2020,Chunaev2021}; see Figure~\ref{fig:scheme_na_to_multiplex} for an illustration.

    \begin{figure}[!htbp]
\centering\begin{minipage}{4.25in}
        \centering
        \begin{minipage}[b]{0.8\linewidth}
            \centering
            \includegraphics[width=0.68\textwidth]{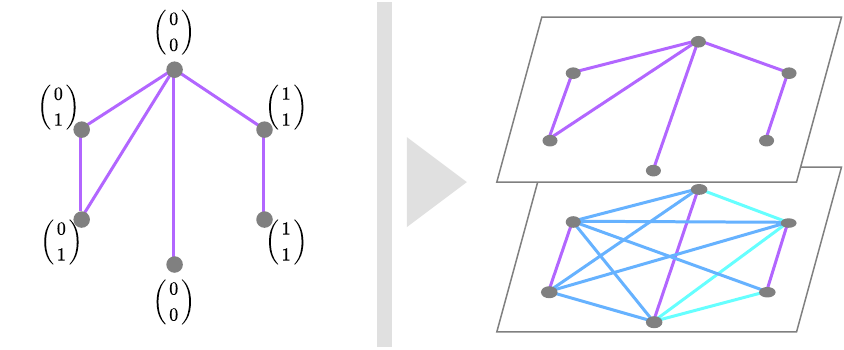}
            \\
            \includegraphics[width=0.05\linewidth]{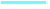}
            link with weight 0 \qquad
            \includegraphics[width=0.05\linewidth]{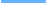}
            link with weight 0.5 \\
            \includegraphics[width=0.05\linewidth]{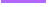}
            link with weight 1
        \end{minipage}
        \caption{From a node-attributed to a multiplex graph. The upper layer is based on the initial network structure. The lower layer is based on the network attributes; the edge weights are the normalized matching coefficients for corresponding pairs of attribute vectors.}
        \label{fig:scheme_na_to_multiplex}
    \end{minipage}
\end{figure}

    \begin{remark}
        Unlike the definition in \cite{Magnani2021}, ours requires a shared set of nodes across layers.
    \end{remark}
    
    \begin{remark}
    A multiplex network is a special case of a {\it multilayer} network. Terminology varies across the literature \cite{Kivela2014, Porter2018, Hanteer2020}; multilayer networks may allow interlayer links and different sets or types of nodes across layers \cite{Snijders2012,Battiston2018}. Multiplex networks may also include interlayer links between copies of the same node. In this paper, the common node set and absence of interlayer edges are assumptions of the model studied, not requirements of every multiplex-network definition.
    \end{remark}

\subsection{Definition of the Community Detection Problem}
\label{sec:cd_definition}

    We define \textit{community detection} (CD) in a multiplex graph $G = {g_l}_{l=1}^L$ as the partitioning of the common node set $\mathcal{V}$ into $K$ nonempty, pairwise disjoint communities. A partition is denoted by
    \begin{equation*}
        C=\{C_k\}_{k=1}^K,
        \label{eq:partition}
    \end{equation*}
    where
    \begin{equation*}
        C_k\subseteq\mathcal{V},\qquad C_k\cap C_{k'}=\emptyset \ \text{for }k\neq k',\qquad \bigcup_{k=1}^K C_k=\mathcal{V}.
    \end{equation*}
    We consider only non-overlapping communities and require every node to belong to exactly one community. Let $\mathcal{C}$ denote the finite set of all such partitions.

    Different community detection methods impose different objective functions on the same set of candidate partitions. We use Newman's modularity~\cite{Newman2006,NewmanPNAS2006} as the common building block for the modularity-based CD methods considered below. We first define modularity for an ordinary weighted graph and a given partition; the {\sf EF}, {\sf SF}, and {\sf LF} objectives introduced in Section~\ref{sec:methods} are then formulated in terms of this single-layer quantity. For a weighted graph and a given partition, we define Newman modularity as follows:

    \begin{equation}
        \label{eq:modularity}
        Q(G,C;\gamma)=\frac{1}{2m} \sum_{k=1}^K \sum_{v_i,v_j \in C_k}  \left[ w(e_{ij}) - \gamma\frac{\kappa(v_i) \kappa(v_j)}{2m}  \right],
    \end{equation}
    where $\kappa(v_i)$ is the weighted degree of vertex $v_i$, $m = \sum_{e_{ij}\in\mathcal{E}}w(e_{ij}) > 0$ is the total weight of $G$ and $\gamma>0$ is a {\it resolution parameter} which enables the detection of communities at different scales~\cite{gamma2006}. Its effects are studied, e.g., in \cite{Lu2020gamma,Newman2016,GadarAbonyi2024}.

    In (\ref{eq:modularity}), the double sum is shorthand for
    \begin{equation*}
        \sum_{v_i,v_j \in C_k} := \sum_{v_i \in C_k} \sum_{v_j \in C_k},
    \end{equation*}
    so that each edge between distinct vertices is counted twice. For an ordinary edge of graphical weight $w$, the symmetric adjacency matrix has $A_{ij}=A_{ji}=w$; for a graphical self-loop of weight $w$, $A_{ii}=2w$. Thus $\kappa(v_i)=\sum_j A_{ij}$ and $m=\frac12\sum_i\kappa(v_i)$. In vertex-indexed sums, weights denote adjacency entries; edge-set sums and construction weights count each graphical edge, including a self-loop, once. LF permits transient self-loops as a bookkeeping device for singleton transient communities.

    \begin{remark}
        The analytical comparison and the historical network experiments concern internal objectives. Section~\ref{sec:cora_external} separately measures external agreement using the adjusted Rand index and normalized mutual information. Internal and external metrics answer different questions~\cite{antonov2022link,Orman2012,Jebabli2018}; annotated categories need not coincide with structural communities~\cite{Hric2014,Newman2015,Peel2017,Chunaev2021}.
    \end{remark}

    Each layer may be assigned a non-negative layer weight $\alpha_l$, which represents its relative importance in the community detection formulation. To make these weights comparable across layers, we normalize the edge weights within each layer so that each layer has the same total edge weight. Without this normalization, differences in the total edge weight of the layers would affect their contributions independently of $\alpha_l$, making the relative importance of a layer difficult to interpret from $\alpha_l$ alone. We therefore require each layer to have a finite, strictly positive total edge weight and normalize its edge weights so that
    \begin{equation}
        \label{eq:weights_normalization}
        \sum_{e_{ij}\in \mathcal{E}} w_l(e_{ij})=1, \quad l=1,\ldots,L.
    \end{equation}
    The layer weights themselves are non-negative and normalized,
    \begin{equation}
        \label{eq:alpha-condition}
        \alpha_l \ge 0 \text{ for } l=1,\ldots,L, \quad \sum_{l=1}^L \alpha_l=1,
    \end{equation}
    so that $\alpha=(\alpha_l)_{l=1}^L$ belongs to the simplex
    \begin{equation*}
        \mathcal{S} = \{\alpha \in [0,1]^L : \sum_{l=1}^L \alpha_l = 1\}.
    \end{equation*}

    We use a common resolution parameter $\gamma>0$ for all layers. The choice $\gamma=1$ recovers standard Newman modularity, while $\gamma<1$ favors larger communities and $\gamma>1$ favors smaller communities.

    \begin{remark}
       Layer-specific resolution parameters $\Gamma={\gamma_l}_{l=1}^L$ are also used in multilayer community detection~\cite{Mucha2010,Hanteer2020}. We restrict our analysis to a common resolution parameter across layers; extending the results to layer-specific resolution parameters is outside the scope of this paper.
    \end{remark}

\section{Related works and classification of CD methods for multiplex networks}
\label{sec:literature}
    A broad range of community detection methods has been developed for multiplex networks. To organize the methods considered in this paper, we adopt a classification based on how they incorporate information from the different layers. Recent reviews and taxonomies \cite{Magnani2021,Hamed2024,BoukabeneMetrics2025} distinguish multiplex CD methods according to how they “handle the presence of multiple layers.” The same early-fusion (EF), simultaneous-fusion (SF), and late-fusion (LF) categories are also used in the node-attributed-network literature \cite{Chunaev2019survey}. Table~\ref{tab:algorithms} summarizes representative methods from both settings and distinguishes modularity-based approaches from methods based on other objectives.

    \begin{table}[h]
        \centering
        \caption{Representative community detection methods illustrating the EF, SF, and LF taxonomy. Methods are grouped by network setting and by the type of objective function.}
        \label{tab:algorithms}
        \small\setlength{\tabcolsep}{2pt}
        \begin{tabular}{|c|c|c|c|c|}
            \hline
            & \multicolumn{2}{c|}{Multiplex networks} &
            \multicolumn{2}{c|}{Node-attributed networks} \\
            \cline{2-5}
            Fusion & Modularity- & Other & Modularity & Other \\
            strategy & based & objectives & based & objectives \\
            \hline
            {EF}
            & \cite{boukabene2024flattening,gurov2022supervised}
            & \cite{Kim2017,Berlingerio2011}
            & \cite{ChunaevComplNetw2020,Chunaev2019,ChunaevYSC2020,Chunaev2021}
            & \cite{Neville2003,Akbas2019,Steinhaeuser2010,DangViennet2012,Ruan2013} \\
            \hline
            {SF}
            & \cite{Mucha2010,Pramanik2026,Pizzuti2017,Zhai2018}
            & \cite{kumar2011co}
            & \cite{Combe2015,CDBNE2022}
            & \cite{Cruz2011,DangViennet2012,Li2018} \\
            \hline
            {LF}
            & \cite{Tagarelli2017,mandaglio2018consensus}
            & \cite{Berlingerio2013,Tang2011,Santra2023,mandaglio2018consensus}
            & \cite{Huang2016}
            & \cite{Luo2019,Liu2020} \\
            \hline
        \end{tabular}
    \end{table}

    Within this taxonomy, we focus on representative modularity-based {\sf EF}, {\sf SF}, and {\sf LF} formulations for shared-node multiplex networks.

    The analytical literature on multilayer and multiplex community detection has considered several related questions. Mucha et al.~\cite{Mucha2010} introduced a generalized multilayer modularity, while Hanteer and Magnani~\cite{Hanteer2020} analyzed community structures recoverable under different interlayer-coupling choices. Other analytical studies have examined the effect of layer aggregation on community detectability: \cite{Taylor2016} studied the detectability of planted communities in layer-aggregated networks using random-matrix theory, and \cite{Taylor2017} extended this analysis to super-resolution detection. Pamfil et al.~\cite{Pamfil2019} relate layer-weighted multilayer modularity to maximum-likelihood inference in stochastic block models. These studies address modularity formulations, detectability, community recovery, or statistical equivalences, rather than explicit analytical relationships among the EF, SF, and LF objectives considered here. The two-layer case $L=2$ with $\gamma=1$ of our {\sf EF}--{\sf SF} relation was previously obtained in \cite{ChunaevComplNetw2020,Chunaev2021}; here, its extension to arbitrary $L$ and $\gamma>0$ is a direct weighted-variance calculation. We use it to compare objectives consistently, characterize the specified {\sf LF} transformation, and distinguish objective gains from external label agreement. Empirical comparisons of multilayer methods have also examined the consequences of link weighting and direction~\cite{SuzukiTsugawa2023}; the present study isolates differences between the stated fusion objectives.

\section{Representative {\sf EF}, {\sf SF}, and {\sf LF} formulations}
\label{sec:methods}

    Following the classification above, we define the specific {\sf EF}, {\sf SF}, and {\sf LF} formulations considered in this study, along with their corresponding modularity-based objectives.

\subsection{{\sf EF} method}
\label{EF_description}

    The {\sf EF} representative generalizes weight-based fusion for CD in node-attributed networks \cite{Chunaev2021,ChunaevComplNetw2020}. It flattens the multiplex network \cite{Magnani2021} by taking a convex combination of the corresponding edge weights across layers; see Figure~\ref{fig:scheme_early}.

    \begin{figure}[h]
\centering\begin{minipage}{4.25in}
        \centering
        \begin{minipage}[b]{\linewidth}
            \centering
            \includegraphics[width=0.85\textwidth]{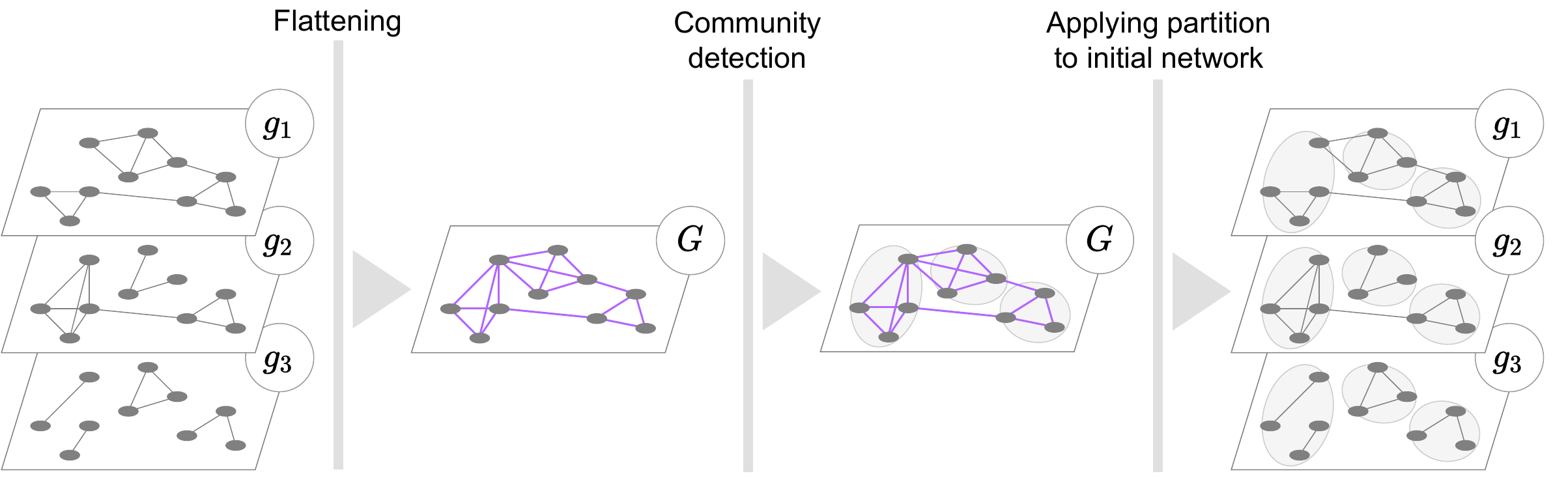}
            \\
            \includegraphics[width=0.05\linewidth]{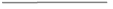}
            link
            \qquad 
            \includegraphics[width=0.05\linewidth]{figures/purple_line.pdf}
            fused link
        \end{minipage}
        \caption{The scheme of the chosen {\sf EF} method.}
        \label{fig:scheme_early}
    \end{minipage}
\end{figure}

    In our notation, the {\sf EF} method fuses the layers of $G$ into the weighted network
    \begin{equation*}
        \hat{G}_{\alpha}=(\mathcal{V},\mathcal{E},\hat{\mathcal{W}}_{\alpha}),
    \end{equation*}
    where the hat denotes quantities associated with the fused network. Its edge weights are given by

    \begin{equation*}
        \label{main_weight}
        \begin{split}
        \hat{w}_{\alpha}(e_{ij})=\sum_{l=1}^L \alpha_l w_l(e_{ij}),\qquad \sum_{l=1}^L \alpha_l = 1.
        \end{split}
    \end{equation*}

    Because each layer is normalized and $\alpha \in \mathcal{S}$, the flattened network is normalized as well:
    
    \begin{equation*}
        \label{wbfm-normalization}
        \sum_{e_{ij}\in \mathcal{E}} \hat{w}_{\alpha}(e_{ij})= 1.
    \end{equation*}

    The resulting CD objective is
    \begin{equation}
        \label{eq:preliminary_Q_EF}
        Q_{\sf EF}(G,C;\gamma,\alpha)=Q(\hat{G}_{\alpha},C;\gamma)
    \end{equation}
    The analytical relation between this objective and the SF objective is derived in Section~\ref{sec:theorems} for arbitrary $L\ge2$ and $\gamma>0$, generalizing the two-layer result of \cite{Chunaev2021, ChunaevComplNetw2020}; see Remark~\ref{EF_remark_L=2}.

\subsection{{\sf SF} method}
\label{SF_description}
    The {\sf SF} representative uses information from all layers simultaneously through the objective~\cite{Magnani2021,Chunaev2019survey}
    
    \begin{equation}
        \label{eq:SF_objective_function_intro}
        Q_{\sf SF}(G,C; \gamma,  \alpha)=\sum_{l=1}^L \alpha_l Q(g_l,C; \gamma),
    \end{equation}

    with $\alpha\in \mathcal{S}$. Thus, for a given layer-weight vector $\alpha$, CD maximizes a convex combination of the single-layer modularity values over $C$ (Figure~\ref{fig:scheme_sim}).
    
    \begin{figure}[h]
\centering\begin{minipage}{4.25in}
        \centering
        \includegraphics[width=0.51\textwidth]{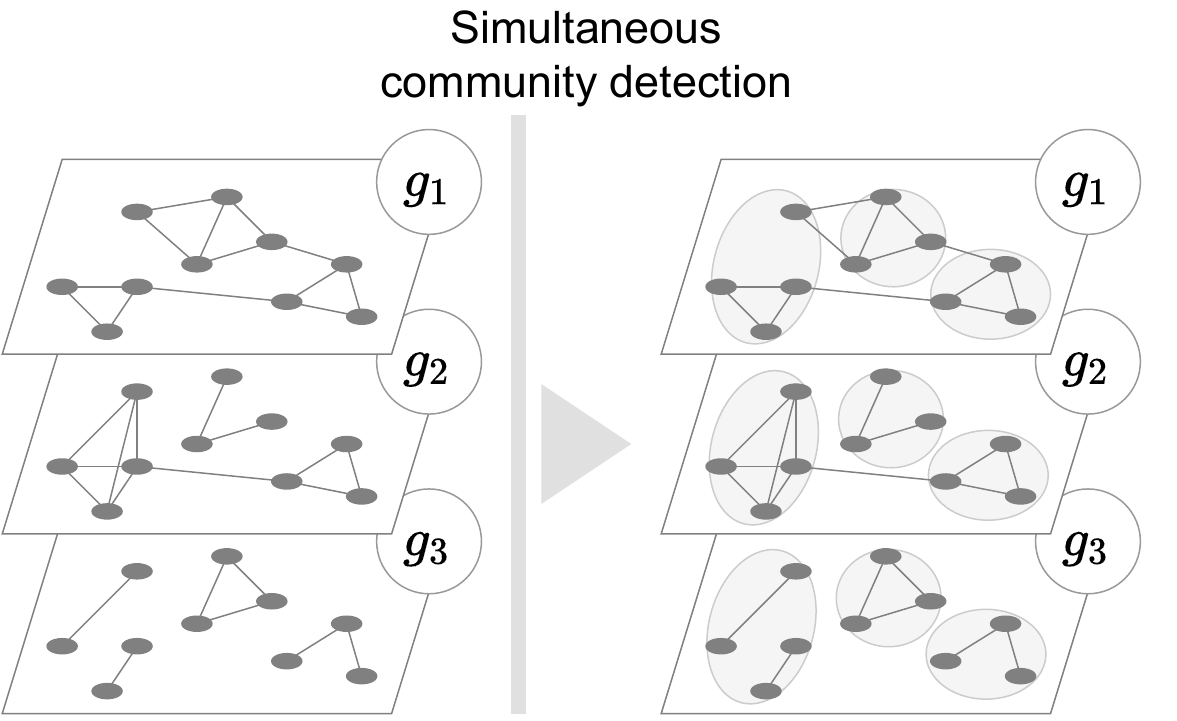}
        \caption{The scheme of the chosen {\sf SF} method.}
        \label{fig:scheme_sim}
    \end{minipage}
\end{figure}

        The {\sf SF} objective is closely related to the generalized multilayer modularity of~\cite{Mucha2010,Hanteer2020}. The generalized multilayer formulation allows interlayer coupling and layer-specific resolution parameters, whereas our setting imposes no interlayer coupling, a common partition across layers, and a common resolution parameter. Under the common-partition assumption, the community-membership factor associated with interlayer coupling between copies of the same node equals $1$.

        Under this common-partition restriction, we write generalized multilayer modularity in our notation as

        \begin{equation*}
        \begin{split}
            \label{eq:Mucha_modularity}
            Q_{\sf gen}(G,C;\Gamma) =& \frac{1}{2\mu} \sum_{l=1}^L \sum_{k=1}^K \sum_{v_i, v_j \in C_k} \left[ w_l(e_{ij}) - \gamma_l \frac{\kappa_l(v_i)\kappa_l(v_j)}{2m_l} \right] \\
            &+ \frac{1}{2\mu} \sum_{l=1}^L \sum_{\phantom{l}s=1\phantom{l}}^L \sum_{v_i \in \mathcal{V}} \mathcal{K}_{sl}(v_i),
            \end{split}
        \end{equation*}
        where $\kappa_l(v_i)$ is the weighted degree of node $v_i$ in layer $l$, $m_l$ is the total edge weight of layer $l$, $\Gamma=\{ \gamma_l \} |_{l=1}^L$ is the set of layer-specific resolution parameters, and $\mathcal{K}_{sl}(v_i)$ denotes the interslice coupling between node $v_i$ in layers $l$ and $s$. The normalization factor is
        
        \begin{equation*}
            \label{mu}
            \mu = \frac{1}{2} \sum_{l=1}^L \sum_{v_i \in \mathcal{V}} \left[ \kappa_l(v_i) + \sum_{s=1}^L \mathcal{K}_{sl}(v_i) \right].
        \end{equation*}

        In our multiplex setting, there is no interlayer coupling, so $\mathcal K_{sl}(v_i)=0$. Hence
        \begin{equation*}
            \label{mu_norm}
            \mu = \frac{1}{2} \sum_{l=1}^L \sum_{v_i \in \mathcal{V}} \kappa_l(v_i) = \frac{1}{2} \sum_{l=1}^L 2 m_l = L,\qquad m_l=1.
        \end{equation*}
        Because each layer is normalized according to Eq.~(\ref{eq:weights_normalization}), $m_l=1$ for every $l$, and therefore $\mu=L$.

        With a common resolution parameter across layers, $\Gamma=\{ \gamma,\ldots,\gamma \}$, the generalized multilayer modularity becomes
    
        \begin{equation*}
            \label{eq:Mucha_modularity_norm}
            Q_{\sf gen}(G,C;\gamma) = \frac{1}{L} \sum_{l=1}^L \frac{1}{2} \sum_{k=1}^K \sum_{v_i, v_j \in C_k} \left[ w_l(e_{ij}) - \gamma\frac{\kappa_l(v_i)\kappa_l(v_j)}{2}  \right].
        \end{equation*}

    By the definition of Newman modularity and the {\sf SF} objective, this is precisely
    
    \begin{equation*}
        \label{eq:Mucha_modularity_norm+}
        Q_{\sf gen}(G,C;\gamma) = Q_{\sf SF}(G,C; \gamma, (\tfrac{1}{L},\ldots,\tfrac{1}{L})).
    \end{equation*}

    A related multi-view formulation by Cruickshank et al.~\cite{cruickshank2020} also uses a weighted sum of resolution-adjusted within-view modularity terms but allows view-specific weights and resolution parameters to be updated iteratively. In contrast, our formulation uses a common resolution parameter, normalized layer weights, and a fixed layer-weight vector $\alpha\in\mathcal{S}$.

\subsection{{\sf LF} method}
\label{LF_description}
    The LF family first processes the layers independently and then combines the resulting community assignments. This general description leaves several choices open, including how communities are detected within each layer and how the layer-wise results are combined. Consequently, different LF methods can yield distinct objectives and analytical properties. In this paper, we consider a specific {\sf LF} construction, partly inspired by~\cite{Liu2020}, in which layer-wise partitions are used to construct transformed transient graphs prior to the final fusion step.

    The first step applies a specified modularity-optimization procedure separately to each layer $g_l$, $l=1,\ldots,L$. Because the layer-wise optimization may admit multiple optimal partitions, the procedure must include a deterministic tie-breaking rule; for heuristic optimization, its algorithmic settings and random seed may also affect the returned partition. We therefore fix a deterministic transient-construction operator $\tau$ that specifies the layer-wise optimization procedure, all settings needed to determine its output, and the graph transformation described below. The layer-wise resolution parameter is included in $\tau$, and $\tau$ is held fixed as $\alpha$ varies. The resulting layer-wise partitions are denoted by
    \begin{equation*}
        C'^{\tau}=\{C_l'^{\tau}\}_{l=1}^L, \text{ where }
        C_l'^{\tau}=\{C_{l,k}'^{\tau}\}_{k=1}^{K_l}.
    \end{equation*}
    We refer to these as the transient partitions. Singleton communities are allowed. A prime denotes quantities associated with the transformed networks constructed from these partitions.

    For transformed graphs, we extend $\mathcal{E}$ to include the singleton self-loops introduced below; original layers have zero weight on these loops. The second step consists of constructing new $L$ {\it transient} networks $g'^{\tau}_l=(\mathcal{V},\mathcal{E},W'^{\tau}_l)$, where $W'^{\tau}_l=\{w'^{\tau}_l(e_{ij})\}$ are based on the set of partitions $C'^{\tau}$. We denote the resulting multiplex by
    \begin{equation*}
        G'^{\tau}=\tau(G)=\{g'^{\tau}_l\}_{l=1}^L.
    \end{equation*}
    Within each transient community, we retain the original edge weights and remove all edges connecting distinct transient communities. For each removed edge, half of its weight is redistributed uniformly among the pairs of distinct vertices in each of its two endpoint communities, including pairs that had zero weight originally. More precisely, for a transient community $B$ with $|B|\ge2$, let $b_l(B)$ denote the total weight of the edges leaving $B$. Then
    \begin{equation*}
        w'^{\tau}_l(e_{ij})=w_l(e_{ij})+\frac{b_l(B)}{|B|(|B|-1)},\qquad v_i,v_j\in B,\quad i\ne j.
    \end{equation*}
    If $C'^{\tau}_{l,k}$ is a singleton ${v_i}$, the corresponding redistributed weight is assigned to the self-loop $e_{ii}$ of the transformed graph and interpreted according to the adjacency convention introduced in Section~\ref{sec:problem_statement}. Specifically,
    \begin{equation*}
        \label{singleton-self-loop}
        w'^{\tau}_l(e_{ii})=\tfrac{1}{2}\,\kappa_l(v_i),
    \end{equation*}
    so that, under the convention $A_{ii}=2w'^{\tau}_l(e_{ii})$, the transformed degree satisfies
    \begin{equation*}
        \kappa'^{\tau}_l(v_i)=2\,w'^{\tau}_l(e_{ii})=\kappa_l(v_i).
    \end{equation*}
    Since each removed edge contributes half of its weight to each endpoint community, the total edge weight is preserved:
    \begin{equation*}
        \sum_e w'^{\tau}_l(e)=1.
    \end{equation*}
    Thus, the transformed layers remain normalized and satisfy Eq.~(\ref{LF_construction_2}), including for singleton communities. Figure~\ref{fig:absorbing} illustrates the construction.

    \begin{figure}[!htbp]
\centering\begin{minipage}{4.25in}
        \centering
        \begin{minipage}[b]{\linewidth}
            \centering
            \includegraphics[width=0.68\textwidth]{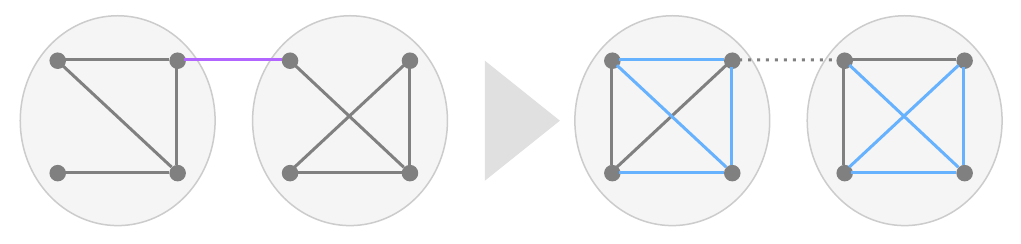}
            \\
            \includegraphics[width=0.05\linewidth]{figures/grey_line.pdf}
            link with weight $\frac{1}{20}$
            \quad
            \includegraphics[width=0.05\linewidth]{figures/blue_line.pdf}
            updated link with weight $\frac{1}{10}$
            \\
            \includegraphics[width=0.05\linewidth]{figures/purple_line.pdf}
            link with weight $\frac{3}{5}$
            \quad
            \includegraphics[width=0.05\linewidth]{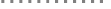}
            erased link
        \end{minipage}
        \caption{Illustration of the ``absorbing'' procedure used to construct the transformed layer network. The construction preserves the layer normalization and satisfies the properties in Eqs.~(\ref{LF_construction_1}) and (\ref{LF_construction_2}).}
        \label{fig:absorbing}
    \end{minipage}
\end{figure}

    By construction, for every layer $l$ and every community $C'^{\tau}_{l,k}$,
    \begin{equation}
        \label{LF_construction_1}
         \sum_{v_i, v_j \in C'^{\tau}_{l,k}} w'^{\tau}_l({e}_{ij}) \ge  \sum_{v_i, v_j \in C'^{\tau}_{l,k}}  w_l({e}_{ij}),
    \end{equation}
    
    \begin{equation}
        \label{LF_construction_2}
        \sum_{v_i \in C'^{\tau}_{l,k}} \kappa'^{\tau}_l(v_i) = \sum_{v_i \in C'^{\tau}_{l,k}} \kappa_l(v_i).
    \end{equation}

    These two properties imply that, for the transient partition $C_l'^{\tau}$, the modularity of the transformed layer cannot be smaller than that of the original layer:
    \begin{equation*}
        Q(g_l'^{\tau},C_l'^{\tau};\gamma) \ge Q(g_l,C_l'^{\tau};\gamma).
    \end{equation*}
    This follows from Lemma~\ref{lemma_lf_connected_to_sf} below. Thus, the construction uses the independently obtained layer-wise partitions to strengthen the within-community structure of the transformed layers. This is a deliberate modeling choice rather than a neutral rewriting of the original networks: redistributing inter-community weight within communities can increase the apparent strength of the corresponding communities.

    The third step applies {\sf EF} to the transient multiplex $G'^{\tau}$ at the chosen $\alpha$, seeking a common partition by maximizing
    \begin{equation}
        \label{preliminary_Q_LF}
        Q_{\sf LF}^{\tau}(G,C;\gamma,\alpha)=Q_{\sf EF}(G'^{\tau},C;\gamma,\alpha).
    \end{equation}
    Section~\ref{sec:theorems} derives its relation to the original layers for fixed $\tau$. The resulting objective depends on both the original multiplex network $G$ and the fixed transient-construction operator $\tau$; Section~\ref{sec:theorems} derives its relation to the original layer objectives. Thus, the analytical results for LF are conditional on the specific construction $\tau$ considered here and do not apply to the entire LF family. When the transient operator is clear from context, we suppress the superscript and write $G'$ and $Q_{\sf LF}$ to avoid overloading the notation. All LF formulas below should therefore be read conditionally on the fixed operator $\tau$.

    \begin{figure}[h]
\centering\begin{minipage}{4.25in}
        \centering
        \begin{minipage}[b]{0.9\linewidth}
            \centering
            \includegraphics[width=0.85\textwidth]{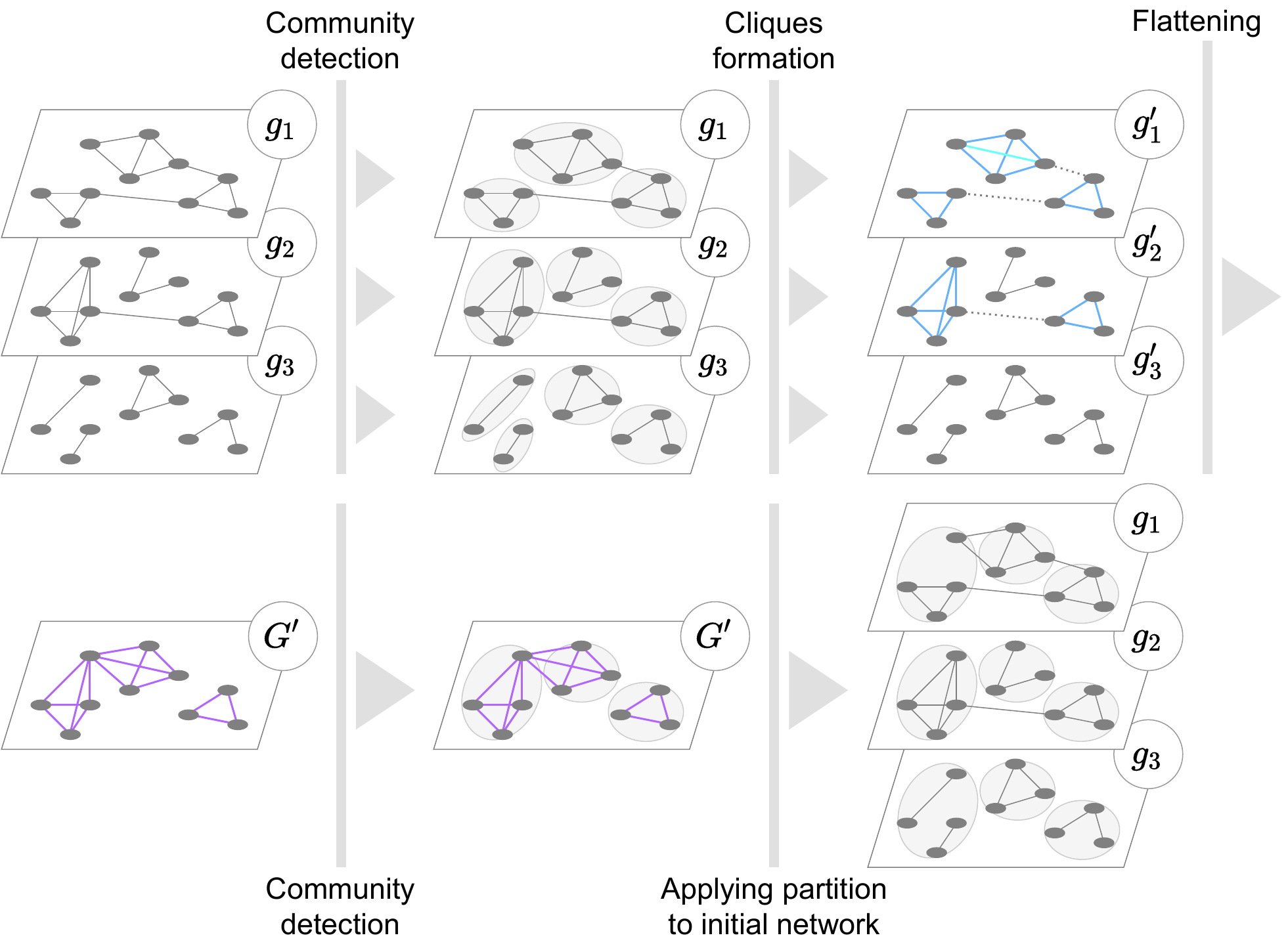}
            \\
            \includegraphics[width=0.05\linewidth]{figures/grey_line.pdf}
            link
            \includegraphics[width=0.05\linewidth]{figures/sky_blue_line.pdf}
            added link
            \includegraphics[width=0.05\linewidth]{figures/blue_line.pdf}
            updated link
            \includegraphics[width=0.05\linewidth]{figures/purple_line.pdf}
            fused link
            \includegraphics[width=0.05\linewidth]{figures/grey_dot_line.pdf}
            erased link
        \end{minipage}
        \caption{The scheme of the chosen {\sf LF} method.}
        \label{fig:scheme_late_new}
    \end{minipage}
\end{figure}

    The overall procedure is illustrated in Figure~\ref{fig:scheme_late_new}. We adopt the absorbing rule because it preserves community-wise degree sums and does not decrease community-wise internal edge weight. Appendix~\ref{app:liu_lf} contrasts this construction with a normalized cluster-network adaptation of \cite{Liu2020}, for which the Eq.~(\ref{LF_construction_1}) need not hold.

\section{Explicit objective functions}
\label{sec:theorems}
    We derive explicit objective functions for the representative EF and LF methods and relate them to SF.

\subsection{Explicit objective function of an {\sf EF} method}
        We derive the explicit {\sf EF} objective in Eq.~(\ref{eq:preliminary_Q_EF}) using the notation in Section~\ref{EF_description}.
    
    \begin{theorem}
        \label{theorem_ef_connected_to_sf}
        Under the normalization in Eq.~(\ref{eq:weights_normalization}), for every $\alpha\in\mathcal{S}$, partition $C\in\mathcal{C}$, and resolution parameter $\gamma>0$, it holds that
        \begin{equation}
            \label{composite_modularity_dependency}
            Q_{\sf EF}(G,C; \gamma,  \alpha)=\sum_{l=1}^L \alpha_l Q(g_l,C; \gamma) +  \gamma\Delta(G,C;\alpha), 
        \end{equation}
        where
        \begin{equation}
            \label{composite_modularity_diff_pos}
            \Delta(G,C;\alpha)= \frac{1}{4}\sum_{k=1}^K \Bigg[\sum_{l=1}^L \alpha_l \bigg[ \sum_{v_i \in C_k} \kappa_l(v_i) \bigg]^2 - \bigg[ \sum_{l=1}^L \alpha_l \sum_{v_i \in C_k} \kappa_l(v_i) \bigg]^2 \Bigg]\geq 0.
        \end{equation}

        Moreover, with $\operatorname{supp}(\alpha)=\{l:\alpha_l>0\}$, we have $\Delta(G,C;\alpha)=0$ if and only if
        \begin{equation}
            \label{condition_jensen_community}
            \sum_{v_i \in C_k} \kappa_{l}(v_i)=\sum_{v_i \in C_k} \kappa_{\lambda}(v_i),\quad \text{for all } k=1,\ldots,K \text{ and all } l,\lambda \in \operatorname{supp}(\alpha),
        \end{equation}
    \end{theorem}

    \begin{equation*}
        d_{lk}(C)=\sum_{v_i\in C_k}\kappa_l(v_i)
        \quad\Longrightarrow\quad
        \Delta(G,C;\alpha)=\frac14\sum_{k=1}^K \operatorname{Var}_{\alpha}\bigl(d_{lk}(C)\bigr),
    \end{equation*}
    where $\operatorname{Var}_{\alpha}$ denotes the weighted variance across layers with weights $\alpha$. Thus, for a fixed partition, the {\sf EF} objective exceeds the {\sf SF} objective by $\gamma/4$ times the sum of the weighted variances of the community-wise degree weights across the layers with $\alpha_l>0$. In other words, {\sf EF} assigns additional objective value to partitions for which the total degree weight of the same community differs substantially between layers.
    
    The proof is given in Appendix~\ref{app:proof_ef}.
    
    \begin{remark}
        \label{EF_remark_L=2}
        Theorem~\ref{theorem_ef_connected_to_sf} is consistent with the result in \cite{ChunaevComplNetw2020,Chunaev2021} for $L=2$ and $\gamma=1$ and thus generalizes it. Note that Theorem~\ref{theorem_ef_connected_to_sf} connects modularities of $L$ layers and modularity of the fused single-layer network whose edge weights are convex combinations of the corresponding layer edge weights.
    
        The magnitude of $\Delta$ represents a measure of layer heterogeneity at the community-degree level. When the active layers have similar community-wise degree sums, $\Delta$ is small; larger differences can make it a substantial component of $Q_{\sf EF}$.
    \end{remark}
    
    \begin{figure}[H]
\centering\begin{minipage}{4.25in}
        \centering
        \includegraphics[width=0.255\linewidth]{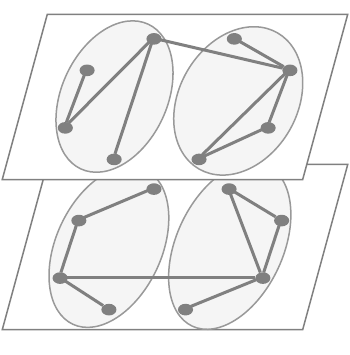}
        \caption{The example in Remark~\ref{remark_8}, where  $\Delta(G,C;\alpha)\equiv 0$ because the corresponding community-wise degree sums agree across layers, while the layers are not identical.}
        \label{fig:example-EF}
    \end{minipage}
\end{figure}

    \begin{remark}
        \label{remark_8}
        Identical layers are sufficient for $\Delta(G,C;\alpha)=0$, but are not necessary (Figure~\ref{fig:example-EF}).
    
        Choose a network $g_1$, a partition $C$, and a permutation of the vertex labels that maps each $C_k$ to itself but changes at least one edge weight. Let $g_2$ be the relabeled layer. The community-wise degree sums are unchanged, so $\Delta(G,C;\alpha)=0$ for every $\alpha$, although the labeled layers differ. Further community-preserving relabelings can be used to construct additional layers. 
    \end{remark}

\subsection{Explicit objective function of an {\sf LF} method}
    We derive the explicit {\sf LF} objective in Eq.~(\ref{preliminary_Q_LF}), conditional on the fixed transient-construction operator $\tau$ defined in Section~\ref{LF_description}.

    The following auxiliary result concerns a normalized single-layer network, an arbitrary test partition, and the transient partition selected by $\tau$.
    \begin{lemma}
        \label{lemma_lf_connected_to_sf}
        Fix $\gamma>0$ and a normalized single-layer network $g=(\mathcal{V},\mathcal{E},\mathcal{W})$. Let $C'$ be the transient partition selected by a fixed operator $\tau$, and let $g'=(\mathcal{V},\mathcal{E},\mathcal{W}')$ be obtained by the absorbing construction in Section~\ref{LF_description}, including its singleton self-loop convention. For any partition $C=\{C_k\}_{k=1}^K$,
        \begin{equation}
            \label{QNewQ}
            \begin{split}
            Q(g',C; \gamma)=Q(g,C; \gamma) + \theta(g,g',C; \gamma),
            \end{split}
        \end{equation}
        
        \begin{equation}
            \label{omega}
            \begin{split}
            \text{where } \theta(g,g',C; \gamma) =
            & \frac{1}{2}\sum_{k=1}^K  \sum_{v_i, v_j \in C_k} \left( w'({e}_{ij}) - w({e}_{ij}) \right)\\
            &-\frac{1}{4}\gamma \sum_{k=1}^K \Bigg[ \bigg[ \sum_{v_i \in C_k} \kappa'(v_i) \bigg]^2 - \bigg[ \sum_{v_i \in C_k} \kappa(v_i) \bigg]^2 \Bigg].
            \end{split}
        \end{equation}

        Moreover, for this transient partition $C'=\{C'_k\}_{k=1}^{K'}$,
        \begin{equation}
            \label{omega-transient}
            \theta(g, g', C'; \gamma) = \frac{1}{2}\sum_{k=1}^{K'}  \sum_{v_i, v_j \in C'_k} \left( w'({e}_{ij}) - w({e}_{ij}) \right) \ge 0,
        \end{equation}
        and $\theta(g, g', C'; \gamma)\equiv 0$ if and only if
        \begin{equation}
            \label{weight_estimation_zero}
            \sum_{v_i, v_j \in C'_{k}} w'(e_{ij})= \sum_{v_i, v_j \in C'_{k}} w(e_{ij}) \qquad \text{for all } k=1,\ldots,K'.
        \end{equation}
    \end{lemma}
    
    The proof is given in Appendix~\ref{app:proof_lf}.

    The following theorem provides an explicit form of the objective function in Eq.~(\ref{preliminary_Q_LF}) within the {\sf LF} method under consideration.

    \begin{theorem}
        \label{theorem_3}
        For any partition $C$, resolution parameter $\gamma>0$, and fixed transient-construction operator $\tau$ as defined in Section~\ref{LF_description}, it holds that
        \begin{equation}
            \label{theorem2_equality}
            \begin{split}
            Q_{\sf LF}^{\tau}&(G,C; \gamma,  \alpha)= \sum_{l=1}^L \alpha_l Q(g'^{\tau}_l,C;\gamma) + \gamma \Delta(G'^{\tau}, C; \alpha)\\
            & = \sum_{l=1}^L \alpha_l Q(g_l,C;\gamma) +  \Theta(G,G'^{\tau},C;\alpha, \gamma) + \gamma \Delta(G'^{\tau}, C; \alpha),
            \end{split}
        \end{equation}
        where
        \begin{equation}
            \label{Theta}
            \Theta(G,G'^{\tau},C;\alpha, \gamma)= \sum_{l=1}^L \alpha_l \theta(g_l, g'^{\tau}_l,C; \gamma).
        \end{equation}
 
        Moreover, $\Delta(G'^{\tau},C;\alpha)\equiv 0$ if and only if $(\ref{condition_jensen_community})$ holds for all $k=1,\ldots,K$ for~$G'^{\tau}$.

        Furthermore, if $g_l\equiv g'_l$ for $l=1,\ldots,L$, then $\theta(g_l,g'_l,C; \gamma)\equiv 0$.
    \end{theorem}

    \begin{proof}
        The former equality in Eq.~(\ref{theorem2_equality}) follows from Eq.~(\ref{preliminary_Q_LF}) and Theorem~\ref{theorem_ef_connected_to_sf}.

        Furthermore, Lemma \ref{lemma_lf_connected_to_sf} immediately gives the latter equality in Eq.~(\ref{theorem2_equality}). The equivalence for $\Delta(G'^{\tau},C;\alpha)$ follows from Theorem~\ref{theorem_ef_connected_to_sf} applied to $G'^{\tau}$. If $g_l=g'_l$, both differences in Eq.~(\ref{omega}) vanish, giving $\theta(g_l,g'_l,C;\gamma)=0$.  
    \end{proof}
    
    \begin{figure}[H]
\centering\begin{minipage}{4.25in}
        \centering
        \includegraphics[width=0.255\linewidth]{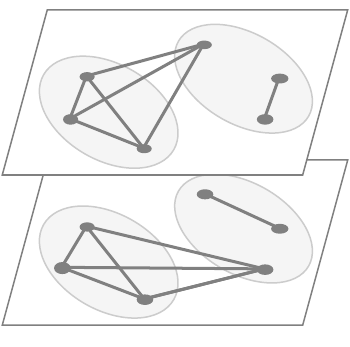}
        \caption{The example in Remark~\ref{remark_9} where the latter two terms in (\ref{theorem2_equality}) are zero simultaneously while the layers are not identical.}
        \label{fig:example-LF-new}
    \end{minipage}
\end{figure}

    \begin{remark}
        \label{remark_9}
        A full structural characterization of when the latter two terms in Eq.~(\ref{theorem2_equality}) vanish simultaneously is outside the scope of this paper. The following example uses a fixed $\gamma=1$ and a fixed transient construction.
        
        The two terms can vanish even when the layers of $G$ differ, as Figure~\ref{fig:example-LF-new} illustrates. Each layer of $G$ consists of two isolated cliques. Assign weight $\tfrac17$ to each of the seven drawn edges in each layer. For a fixed operator $\tau$ selecting the two connected components as transient communities in each layer, $G'^{\tau}=G$ and thus $\theta(g_l, g'^{\tau}_l,C;\gamma)\equiv 0$, $l=1,2$, for any $C$. Consider the partition of $G'^{\tau}$ indicated in Figure~\ref{fig:example-LF-new} (in fact, the node numbers here are just rearranged within the communities as in Remark~\ref{remark_8}). The two community volumes of this test partition are $9/7$ and $5/7$ in both layers, so Jensen's inequality in Eq.~(\ref{condition_jensen_community}) is satisfied. Thus, the latter two terms in Eq.~(\ref{theorem2_equality}) are zero simultaneously, although the two layers are distinct as labeled graphs.
    \end{remark}

\subsection{Conclusions on the explicit objective functions}
    Here we summarize our theoretical results about the objective functions in Theorems~\ref{theorem_ef_connected_to_sf} and~\ref{theorem_3}. By means of Eqs.~(\ref{eq:preliminary_Q_EF}), (\ref{eq:SF_objective_function_intro}) and (\ref{preliminary_Q_LF}) we have:
    
    \begin{equation*}
        \label{conclusion_formulas}
        \begin{split}
        Q_{\sf EF}(G, C; \gamma,\alpha)&= Q_{\sf SF}(G, C; \gamma,\alpha) + \gamma \Delta(G,C;\alpha) \\
        Q_{\sf LF}^{\tau}(G, C; \gamma,\alpha) &= Q_{\sf EF}(G'^{\tau}, C; \gamma,\alpha)\\
        &= Q_{\sf SF}(G, C; \gamma,\alpha) + \Theta(G,G'^{\tau},C;\alpha,\gamma) + \gamma \Delta(G'^{\tau}, C;\alpha)\\
        &= Q_{\sf EF}(G, C; \gamma,\alpha) - \gamma\Delta(G,C;\alpha) \\
        &\quad + \Theta(G,G'^{\tau},C;\alpha,\gamma) + \gamma\Delta(G'^{\tau},C;\alpha).
        \end{split}
    \end{equation*}
    
    These identities give explicit objectives for the representative {\sf EF} and {\sf LF} methods. Their different data-transformation stages still distinguish them from {\sf SF}. The following consequences hold:
    \begin{itemize}
        \item for fixed $C$ and $\gamma>0$, $Q_{\sf EF}=Q_{\sf SF}$ exactly when $\Delta(G,C;\alpha)=0$, and $Q_{\sf LF}^{\tau}=Q_{\sf SF}$ exactly when $\Theta(G,G'^{\tau},C;\alpha,\gamma)+\gamma\Delta(G'^{\tau},C;\alpha)=0$. All three values agree when both conditions hold, which does not require identical layers (Remarks~\ref{remark_8} and~\ref{remark_9}). Equality at one partition does not make the objective functions identical. Conversely, distinct objective functions can share the same argmax set;
        \item $\Delta(G,C;\alpha)$ is non-negative and can be large \cite{ChunaevComplNetw2020,Chunaev2021}. Lemma~\ref{lemma_lf_connected_to_sf} guarantees $\theta(g_l,g'^{\tau}_l,C;\gamma)\ge0$ for the transient partition $C=C'^{\tau}_l$, not for an arbitrary partition. LF is evaluated on the modified layers, so its absolute values should not be ranked against EF and SF values;
        \item the explicit objectives specify the criterion for a \textit{community} adopted by each method (also noted by \cite{Magnani2021}). Evaluating a returned partition under the optimized objective measures optimization quality. External or auxiliary criteria assess different properties and can be appropriate for the scientific task. Our comparisons concern internal objective values; their relation to external quality criteria is outside the scope of this paper.
    \end{itemize}

    Choose one maximizer of each objective and write
    \begin{align}
        C_{\sf EF}&\in\arg\max_C Q_{\sf EF}(G, C; \gamma,\alpha),\\
        C_{\sf SF}&\in\arg\max_C Q_{\sf SF}(G, C; \gamma,\alpha),\\
        C_{\sf LF}^{\tau}&\in\arg\max_C Q_{\sf LF}^{\tau}(G, C; \gamma,\alpha).
    \end{align}
    By definition, each optimizer dominates the competing partitions under its own objective. This bookkeeping observation is tautological and should not be interpreted as a substantive cross-method ranking on a common scale. The meaningful comparison is instead cross-objective: for a partition returned by one method, how do the other objective functions evaluate it? We adopt this interpretation in the experimental section.

    Consequently, if a heuristic optimizer returns a partition whose cross-objective ranking differs from the exact argmax bookkeeping above, this does not contradict the theory. It only reflects optimizer dependence.

\section{Finding optimal \texorpdfstring{$\alpha$}{alpha}}
\label{sec:synergy}
    For a fixed partition $C$, the dependence of the three objective functions on $\alpha$ has a simple geometric interpretation. The {\sf SF} objective is affine in $\alpha$, whereas the {\sf EF} objective is a concave quadratic perturbation of the same affine term. Consequently, optimal layer weights behave differently for {\sf SF} and {\sf EF}. The following elementary corollary records the standard simplex argument for reference.

    \begin{corollary}
        \label{theorem:synergy_SF}
        For any multiplex graph $G$ and any $\gamma>0$,
        \begin{equation}
            \max_{\alpha \in \mathcal{S},\, C \in \mathcal{C}} Q_{\sf SF}(G, C; \gamma, \alpha)
            =
            \max_{\alpha \in \mathcal{S}_0,\, C \in \mathcal{C}} Q_{\sf SF}(G, C; \gamma, \alpha),
        \end{equation}
        where $\mathcal{S}_0$ is the set of simplex vertices, each assigning weight 1 to one layer and 0 to all others. In particular, there exists $\alpha^*\in\mathcal{S}_0$ that attains the global maximum.
    \end{corollary}

    \begin{proof}
        Fix a partition $\tilde{C}\in\mathcal{C}$. Then
        \begin{equation}
            Q_{\sf SF}(G,\tilde{C};\gamma,\alpha)=\sum_{l=1}^L \alpha_l Q(g_l,\tilde{C};\gamma),
        \end{equation}
        which is affine in $\alpha$ on the simplex $\mathcal{S}$. Therefore, its maximum over $\mathcal{S}$ is attained at a vertex of $\mathcal{S}$. Maximizing over all $\tilde{C}\in\mathcal{C}$ preserves this property.
    \end{proof}

    \begin{remark}
        Corollary~\ref{theorem:synergy_SF} guarantees a vertex maximizer, not uniqueness. For a fixed globally optimal partition $C^*$, any convex combination of vertices whose layers attain the global maximum at $C^*$ is also optimal. These weights form a face of $\mathcal{S}$ for that partition; different optimal partitions need not give the same face.
    \end{remark}

    \begin{theorem}
        \label{theorem:ef_fixed_partition}
        Fix a multiplex graph $G$, a partition $C$, and a resolution parameter $\gamma\ge 0$, using the continuous extension of modularity in Eq.~(\ref{eq:modularity}) at $\gamma=0$. Then $Q_{\sf EF}(G,C;\gamma,\alpha)$ is a concave function of $\alpha$ on the simplex $\mathcal{S}$. In particular, an {\sf EF} optimum may occur on the boundary or in the interior of the simplex.

        For $L=2$, writing $\alpha=(a,1-a)$ with $a\in[0,1]$, one has
        \begin{equation}
            \label{eq:ef_two_layer_alpha}
            \begin{split}
            Q_{\sf EF}(G,C;\gamma,(a,1-a))
            &=
            aQ(g_1,C;\gamma)+(1-a)Q(g_2,C;\gamma)
            \\
            &+\frac{\gamma}{4}a(1-a)\sum_{k=1}^K \big(d_{1k}(C)-d_{2k}(C)\big)^2,
            \end{split}
        \end{equation}
        where $d_{lk}(C)=\sum_{v_i\in C_k}\kappa_l(v_i)$.
    \end{theorem}

    The proof is given in Appendix~\ref{app:proof_concavity}.

    The profile $\alpha\mapsto\max_C Q_{\sf EF}(G,C;\gamma,\alpha)$ is the pointwise maximum of fixed-partition concave quadratics and is therefore not necessarily concave globally.

    \begin{remark}
        \label{remark:ef_interior_counterexample}
        The {\sf EF} optimum over $\alpha$ need not lie at a simplex vertex. Consider a two-layer multiplex on nodes $\{1,2,3,4\}$ with one unit-weight edge $(1,2)$ in layer $1$ and one unit-weight edge $(3,4)$ in layer $2$. Let $\gamma=1$. For each individual layer, the maximum modularity is $0$, so every simplex vertex gives value $0$. However, for $a=\tfrac12$ and partition $\{\{1,2\},\{3,4\}\}$, the fused graph has two disjoint edges of weight $\tfrac12$ and modularity $\tfrac12$. Hence
        \begin{equation*}
            \max_{\alpha\in\mathcal{S},\,C} Q_{\sf EF}(G,C;\gamma,\alpha)
            >
            \max_{\alpha\in\mathcal{S}_0,\,C} Q_{\sf EF}(G,C;\gamma,\alpha).
        \end{equation*}
        Nontrivial fusion can therefore be optimal for the {\sf EF} objective.
    \end{remark}

    For LF, no equally general vertex theorem is available at the level of the original graph $G$. Once the transient multiplex $G'$ is fixed, the final {\sf LF} optimization inherits the {\sf EF} geometry on $G'$. However, the transient construction itself depends on the returned layer-wise partitions and on the absorbing rule, so LF remains more construction-dependent than SF or EF.

\subsection{Interaction with the modularity resolution limit}
\label{sec:resolution_limit}

    The three formulations retain the resolution-limit concern of modularity \cite{Fortunato2007,Lancichinetti2011}. Multiplying all of a layer's edge weights by a positive constant leaves its modularity unchanged, so normalization to total weight one does not remove this limitation or make different layers share a detection scale. Fusion changes the effective weights or objective, and its effect on detectable communities depends on the layer and community structure. Our objective identities do not provide a separate resolution-limit theorem for the fused networks.

    These observations are not peculiar to our setting: the resolution limit is a property of the modularity objective itself, and the same caveat would apply to any modularity-based EF, SF, or LF construction. Resolution-limit-free alternatives such as the constant-Potts model of \cite{Traag2011CPM} replace the null model and lie outside the scope of this paper, but adapting the EF--SF--LF decomposition to those objectives is a natural direction for future work.

\section{Experimental study}
\label{sec:experiments}
    In this section, we illustrate theoretical conclusions regarding the objective functions of the {\sf EF}, {\sf SF}, and {\sf LF} methods. For simplicity, we set $\gamma=1$ everywhere below. The exact statements in Section~\ref{sec:synergy} therefore apply to the toy examples and to the objective functions themselves.

    We use a two-level experimental design. First, we consider brute-force optimization for toy multiplex networks in Section~\ref{exp_toy}. These experiments are exact and therefore isolate the geometry of the objective functions from optimizer effects. Second, we study real-world node-attributed and multiplex networks in Section~\ref{exp_real}. Those experiments rely on Louvain and Leiden heuristics, so they should be interpreted as illustrations of optimizer dependence rather than as tests of exact argmax theory. Section~\ref{sec:cora_external} then evaluates whether objective-based weight selection agrees with external annotations on label-free Cora.

\subsection{Toy examples of multiplex networks and brute-force optimization}
    \label{exp_toy}
    \begin{figure}[!htbp]
\centering\begin{minipage}{4.25in}
        \centering
        \begin{subfigure}[b]{0.49\textwidth}
            \centering
            \centering\includegraphics[width=0.357\linewidth]{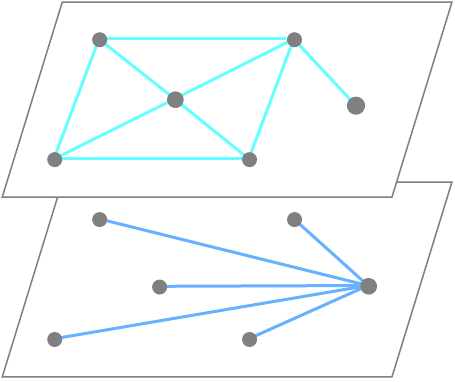}
           \caption{ }
            \label{fig:BF1_graph}
        \end{subfigure}
        \begin{subfigure}[b]{0.49\textwidth}
            \centering
            \centering\includegraphics[width=0.357\linewidth]{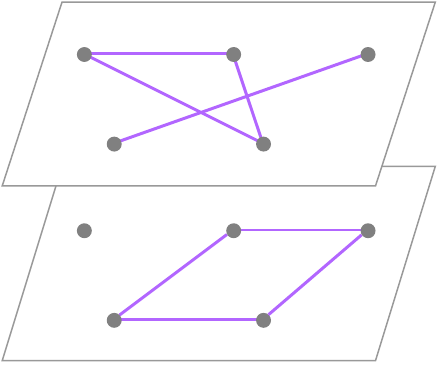}
            \caption{ }
            \label{fig:BF2_graph}
        \end{subfigure}
        \begin{subfigure}[b]{0.8\textwidth}
            \centering
            \includegraphics[width=0.05\linewidth]{figures/sky_blue_line.pdf}
            link with weight $\frac{1}{9}$
            \hfill
            \includegraphics[width=0.05\linewidth]{figures/blue_line.pdf}
            link with weight $\frac{1}{5}$
            \hfill \\
            \includegraphics[width=0.05\linewidth]{figures/purple_line.pdf}
            link with weight $\frac{1}{4}$
        \end{subfigure}
        \caption{Small two-layer multiplex networks for the brute force experiments: (a) Example 1, (b) Example 2. The examples have six and five nodes, respectively, allowing exhaustive enumeration of 203 and 52 set partitions.}
        \label{brute-force-examples}
    \end{minipage}
\end{figure}

    \begin{figure}[!htbp]
\centering\begin{minipage}{4.25in}
        \centering
    \scalebox{1.0}{%
    \begin{minipage}{\textwidth}
    \centering
        \begin{subfigure}[b]{0.49\textwidth}
            \centering\includegraphics[width=\linewidth]{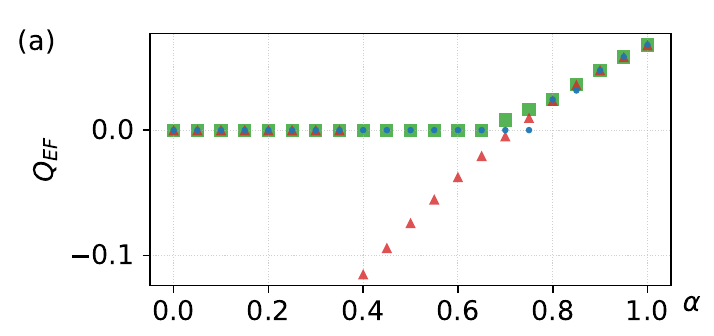}
            \phantomsubcaption\label{fig:toy1_ef}
        \end{subfigure}
        \begin{subfigure}[b]{0.49\textwidth}
            \centering\includegraphics[width=\linewidth]{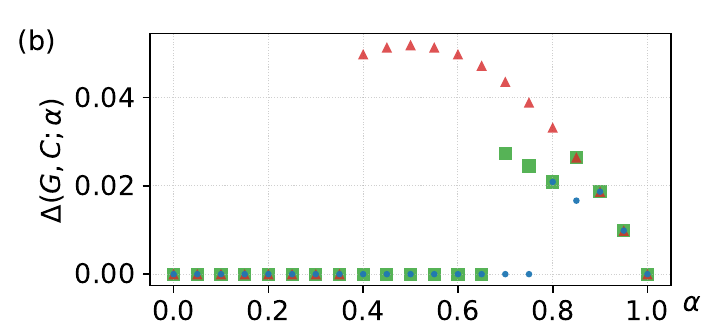}
            \phantomsubcaption\label{fig:toy1_delta}
        \end{subfigure}
        \begin{subfigure}[b]{0.49\textwidth}
            \centering\includegraphics[width=\linewidth]{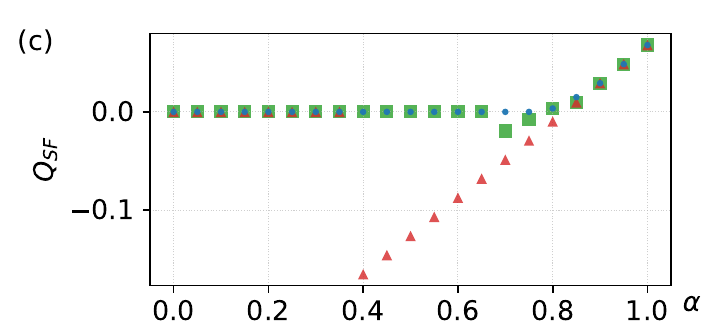}
            \phantomsubcaption\label{fig:toy1_sf}
        \end{subfigure}
        \begin{subfigure}[b]{0.49\textwidth}
            \centering\includegraphics[width=\linewidth]{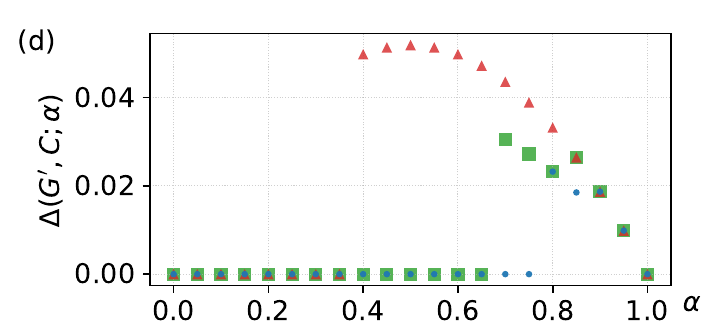}
            \phantomsubcaption\label{fig:toy1_delta_new}
        \end{subfigure}
        \begin{subfigure}[b]{0.49\textwidth}
            \centering\includegraphics[width=\linewidth]{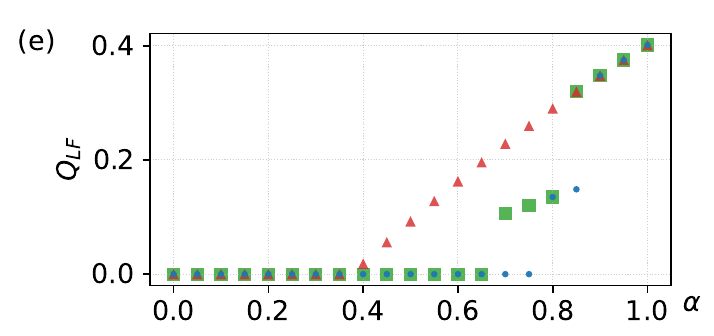}
            \phantomsubcaption\label{fig:toy1_lf}
        \end{subfigure}
        \begin{subfigure}[b]{0.49\textwidth}
            \centering\includegraphics[width=\linewidth]{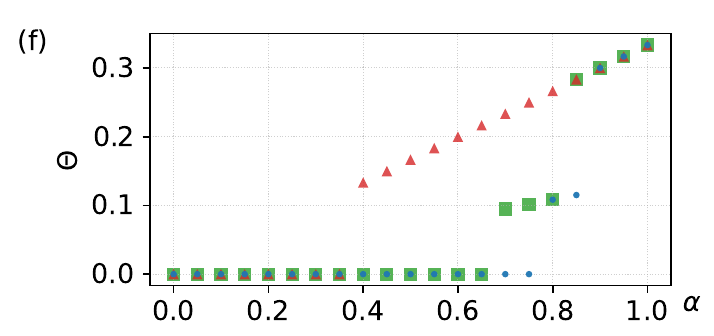}
            \phantomsubcaption\label{fig:toy1_theta}
        \end{subfigure}
    
    \toymethodlegend
    \end{minipage}%
    }
    \caption{Toy Example 1. Brute force results of the {\sf EF}, {\sf SF} and {\sf LF} methods (producing correspondingly the partitions $C_{\sf EF}$, $C_{\sf SF}$ and $C_{\sf LF}$) in terms of the functions $Q_{\sf EF}$, $Q_{\sf SF}$, $Q_{\sf LF}$, $\Delta(G,C;\alpha)$, $\Delta(G'^{\tau},C;\alpha)$ and $\Theta$ for $\alpha=(a,1-a)$ with $a\in\{0.00,0.05,\ldots,1.00\}$. The vertical axis of the $Q_{\sf LF}^{\tau}$ panel is not on the same raw scale as $Q_{\sf EF}$ and $Q_{\sf SF}$ because $Q_{\sf LF}^{\tau}$ is evaluated on the transient layers $G'^{\tau}$; only within-panel comparisons of $Q_{\sf LF}$ values are meaningful across $\alpha$.}
        \label{brute_force_example_result_1}
    \end{minipage}
\end{figure}

    The toy examples demonstrate that, for some $G$ and $\alpha$, the exact cross-objective scores satisfy
    \begin{equation}
        \label{quality_hierarchy_more}
        \begin{split}
        &Q_{\sf EF}(G, C_{\sf EF}; \gamma, \alpha)>Q_{\sf EF}(G, C_{\sf SF}; \gamma, \alpha)>Q_{\sf EF}(G, C_{\sf LF}; \gamma, \alpha),\\
        &Q_{\sf SF}(G, C_{\sf SF}; \gamma, \alpha)>Q_{\sf SF}(G, C_{\sf EF}; \gamma, \alpha)>Q_{\sf SF}(G, C_{\sf LF}; \gamma, \alpha),\\
        &Q_{\sf LF}^{\tau}(G, C_{\sf LF}; \gamma, \alpha)>Q_{\sf LF}^{\tau}(G, C_{\sf EF}; \gamma, \alpha)> Q_{\sf LF}^{\tau}(G, C_{\sf SF}; \gamma, \alpha).
        \end{split}
    \end{equation}

    Here, each row uses one objective function to score the three partitions returned by EF, SF, and LF. These strict rankings hold for particular examples and weights, not for all networks and $\alpha$.

    For each grid point, we exhaustively enumerate all set partitions. If several partitions maximize EF or SF, we retain the first in enumeration order. For LF, we consider every combination of maximizing layer partitions; each combination defines a transient graph, for which we retain the first final LF maximizer in enumeration order. The figures show all these transient-graph choices, but not every tied final partition. We evaluate the retained partitions under each objective at $\alpha=(a,1-a)$ for $a\in\{0.00,0.05,\ldots,1.00\}$.

    Figure~\ref{brute-force-examples}a has a star in the lower layer and a denser upper layer, giving different community-wise degree sums across layers and illustrating the contribution of $\Delta$. In Figure~\ref{brute_force_example_result_1}, the EF- and SF-returned partitions have different cross-objective scores precisely at $a=0.70,0.75,0.85$, where their $\Delta$ values differ.

    At the tested weights $a\in\{0.00,0.05,\ldots,0.35\}\cup\{0.90,0.95,1.00\}$, the three methods return partitions with coincident scores under each objective. At other tested weights, at least one cross-objective score differs. In particular, the strict hierarchy (\ref{quality_hierarchy_more}) holds for $a=0.70$. At $a=0.75$, the LF partition scores above the SF partition under $Q_{\sf EF}$, so that hierarchy does not hold.

    Figure~\ref{brute-force-examples}b has two cliques in the upper layer and a structurally different lower layer, with similar degree distributions. Appendix~\ref{app:toy_detail} shows the corresponding objective profiles (Figure~\ref{brute_force_example_result_2}) and transient partitions (Figure~\ref{fig:transient_diversity}); we discuss their implications here.

    This example shows how tied layer-wise maximizers can produce different LF objectives and final partitions. If layer $l$ has $T_l$ transient maximizers, there are $\prod_{l=1}^L T_l$ combinations of transient partitions; different combinations need not yield distinct transient graphs or final solutions. Figure~\ref{fig:transient_diversity} shows all lower-layer transient partitions for the network in Figure~\ref{brute-force-examples}b; each attains the maximum layer modularity of $0$. Absorption produces different transient graphs, yielding multiple LF solutions under all objectives. Because evaluating $Q_{\sf LF}$ also requires a transient graph, every source method can have multiple $Q_{\sf LF}$ scores (Figure~\ref{brute_force_example_result_2}).

    At $a=0.55$, some LF branches coincide with the EF optimum, while others score lower under $Q_{\sf EF}$ and $Q_{\sf SF}$. Thus the strict hierarchy (\ref{quality_hierarchy_more}) does not hold for every transient choice. Coincident branches at other grid points likewise do not imply uniqueness. Since the transient partitions are all layer-wise modularity maximizers, the spread among their induced LF values cannot be ignored: for this network it is largest at $a=0$, where all weight sits on the lower layer (which has eight transient maximizers), and decreases monotonically to zero at $a=1$, where all weight sits on the upper layer (whose transient maximizer is unique). This monotone collapse is specific to the transient-multiplicity structure of this example rather than a general property of {\sf LF}.
    
    These exact toy experiments illustrate the main theoretical messages of Section~\ref{sec:synergy}. For each fixed partition, SF is affine and EF is concave in $\alpha$; their optimized profiles may switch between partitions, and EF may prefer nontrivial fusion. The LF results depend strongly on the multiplicity of transient partitions, which is precisely the construction-dependence emphasized above.

\subsection{Heuristic modularity optimization}
\label{exp_real}
    We compare Louvain \cite{Blondel2008,Yang2016LFR} and Leiden \cite{Traag2019} on real-world networks. Both algorithms are heuristics and need not attain a global maximum of the tested objectives.

    The two-layer grid varies $\alpha_1\in[0,1]$ with step $0.05$ and $\alpha_2=1-\alpha_1$ (the plots abbreviate $\alpha_1$ as $\alpha$); three-layer experiments use the corresponding simplex grid. Each grid point has 30 restarts with seeds $\{0,\ldots,29\}$, and Leiden uses two optimization iterations per restart. Two-layer figures show median trajectories with ribbons spanning the first to third quartiles (the interquartile range, IQR); three-layer figures show median contour maps. We also report exploratory seed-resampling intervals for winning medians and winner gaps, runtimes and community counts.

    \noindent\textbf{Reporting protocol.} For each dataset, optimizer, source method, target objective, and grid point, we summarize the 30 restart scores by their median and IQR. Restart dispersion is measured by mean pairwise normalized variation of information,
    \begin{equation}
        \operatorname{NVI}(C,D)=\frac{H(C)+H(D)-2I(C;D)}{\log |\mathcal V|},
    \end{equation}
    where $C$ and $D$ are partitions of the same node set, $H$ is Shannon entropy of their community-size proportions, and $I(C;D)$ is mutual information between their community assignments, using natural logarithms and the joint proportions $|C_k\cap D_j|/n$, with zero-probability terms contributing zero. The denominator is positive because $n\ge2$. For a fixed dataset, optimizer, and target objective, a configuration is a source method together with a grid point $\alpha$; the winner and runner-up are ranked by their median scores, with the best restart score breaking median ties. The winner interval uses the 2.5th and 97.5th percentiles of medians from 2000 bootstrap resamples of its 30 restart scores. The winner-gap interval uses the same percentiles of the difference between separately resampled winner and runner-up medians. The two configurations are selected on these same runs and resampled independently rather than as seed-paired observations. These are conditional descriptive summaries of seed variability, not confidence intervals for selection or population effects; they do not quantify uncertainty over datasets, optimizers, or grid choices. The runner-up may use the same $\alpha$ with a different source method, so this interval does not establish that the winning layer weights are distinct from other grid points. All reported ``best'' values are grid-restricted best medians under the tested optimizer settings, not exact optima.

    The code repository and artifact release status are described in the Data and code availability statement.

\subsubsection{Historical network comparisons}
We use undirected, cleaned, shared-node versions of Political Blogs, Cora, London Stations, Major Airlines, and the two Immune Trials networks. Appendix~\ref{app:historical} retains the dataset definitions, preprocessing details, full objective profiles, and numerical summaries. The historical Cora attribute representation includes subject categories and therefore does not independently validate subject recovery; Section~\ref{sec:cora_external} evaluates a separate word-only construction.

For each two-layer comparison, all three fusion methods use the same optimizer family. EF and LF follow Sections~\ref{EF_description} and \ref{LF_description}; LF uses that optimizer for both transient and final partitions. The three-layer experiments use Leiden for all three methods. For LF cross-objective evaluation, the score $Q_{\sf LF}$ is computed on the transient graph $G'$ produced for the same dataset, optimizer, and seed and reused across the $\alpha$ grid; hence LF scores are conditional on that transient construction rather than being a single seed-independent deterministic scale across all runs.

Figure~\ref{fig:london_stations_target} illustrates the interior EF winner on London Stations. The other historical networks attain their highest observed median EF score at a simplex vertex. The application-specific interpretation and limits of these comparisons are discussed in Section~\ref{sec:discussion}.

\begin{figure}[!htbp]
\centering\begin{minipage}{4.25in}
    \centering
    \scalebox{1.0}{%
    \begin{minipage}{\textwidth}
    \centering
    \begin{subfigure}[b]{0.98\textwidth}
        \centering\includegraphics[width=\linewidth]{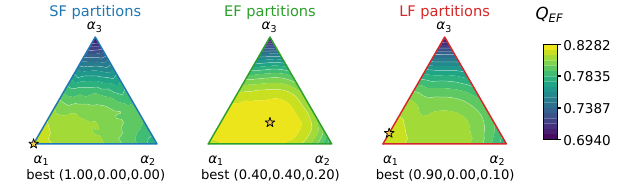}
    \end{subfigure}
    \vspace{0.6em}
    \begin{subfigure}[b]{0.98\textwidth}
        \centering\includegraphics[width=\linewidth]{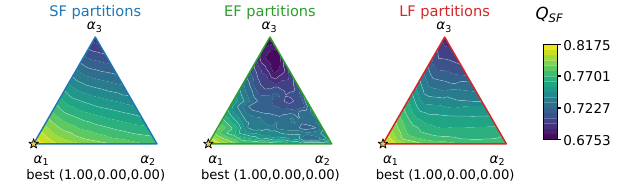}
    \end{subfigure}
    \vspace{0.6em}
    \begin{subfigure}[b]{0.98\textwidth}
        \centering\includegraphics[width=\linewidth]{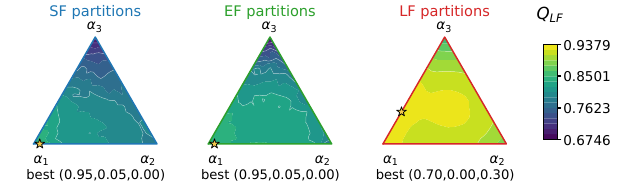}
    \end{subfigure}

    \end{minipage}%
    }
    \caption{London Stations multiplex network modularity optimization performed by Leiden-based algorithms. DLR denotes the Docklands Light Railway. Target functions. The $Q_{\sf LF}^{\tau}$ row is evaluated on the transient layers $G'^{\tau}$ and is not on the same raw scale as $Q_{\sf EF}$/$Q_{\sf SF}$. Layer order is $\alpha_1$: Tube, $\alpha_2$: Overground, $\alpha_3$: DLR.}
    \label{fig:london_stations_target}
\end{minipage}
\end{figure}

\subsection{External validation on label-free Cora}
\label{sec:cora_external}

To assess whether objective values translate into agreement with external annotations, we repeat the Cora experiment with subject categories excluded from every graph-construction and optimization step. The citation layer is unchanged. The attribute layer uses only the 1433 binary word indicators, cosine similarity, the same top-25-neighbor union rule, and unit total edge weight. It contains 2708 nodes and 49319 undirected edges; the citation layer contains 5278 edges. The seven subject categories are retained only for post-optimization evaluation. This experiment is distinct from the historical Cora results in Appendix~\ref{app:historical} and Appendix~\ref{app:cora_sensitivity}, which used category features.

We retain the historical optimization settings: $\gamma=1$, $\alpha=(a,1-a)$ with $a\in\{0,0.05,\ldots,1\}$, seeds 0--29, both Louvain and Leiden, and all three fusion methods. The first weight refers to citations. The LF first-stage seed is $s+1000$ for citations and $s+1001$ for words, where $s$ is the restart seed; Leiden uses two iterations at both stages. No setting is selected using subject labels. For each optimizer and source method, the reported objective-selected weight maximizes the median of that method's own objective over the 30 seeds; ties within $10^{-12}$ select the smaller $a$. We also report the citation-only ($a=1$), words-only ($a=0$), and equal-weight ($a=0.5$) configurations. At an endpoint, EF and SF operate on an original single layer, whereas LF still operates on a transformed layer.

The primary external metric is the adjusted Rand index (ARI)~\cite{HubertArabie1985}. The secondary metric is normalized mutual information (NMI), using the arithmetic-mean entropy normalization $2I(C;Y)/(H(C)+H(Y))$, where $Y$ denotes subject categories. Both metrics match node IDs exactly. We report medians and interquartile ranges across seeds, together with the number of detected communities. These describe optimization variability on one fixed network, not population uncertainty or evidence of general superiority. The analysis contract, including metrics, weight selection, and baselines, was fixed before inspecting external scores.

Table~\ref{tab:cora_external_baselines} in Appendix~\ref{app:cora_detail} gives the three fixed baselines; its citation-only rows also report the objective-selected configurations. Figure~\ref{fig:cora_external} shows the full weight profiles.

All six method--optimizer combinations select the citation-only weight $a=1$ by their own objective. Nevertheless, the fixed equal-weight configuration has a higher median ARI than both that choice and the words-only baseline for every combination. Across the six combinations, equal-weight median ARI ranges from 0.381 to 0.448, compared with 0.228--0.259 at the objective-selected weights and 0.301--0.311 for words only. NMI gives the same direction of median comparison, although the magnitude differs across methods and optimizers. These are descriptive comparisons on one network, not significance tests. They show that the preference of an internal fusion objective for a single layer need not identify the configuration with better external agreement: incorporating word information improves subject agreement here even though it does not maximize the tested internal objective.

Detailed median and quartile values, including community counts, are retained in Appendix~\ref{app:cora_detail}.

\begin{figure}[htbp]
\centering\begin{minipage}{4.25in}
\centering
\includegraphics[width=3.6in]{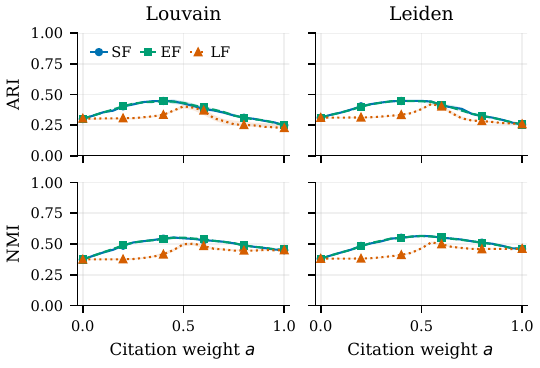}
\caption{External agreement across citation weights on label-free Cora. Curves show the median ARI (top) and NMI (bottom) over 30 seeds; shaded bands span the first and third quartiles. Columns distinguish Louvain and Leiden, while colors, line styles and markers distinguish the source fusion methods. Subject categories are used only for evaluation.}
\label{fig:cora_external}
\end{minipage}
\end{figure}

\section{Discussion and conclusions}
\label{sec:discussion}

    Section~\ref{sec:introduction} emphasized that empirical rankings and descriptions of community detection methods do not necessarily reveal what their fusion objectives optimize. For the representative {\sf EF}, {\sf SF}, and {\sf LF} formulations studied here, the identities derived in Section~\ref{sec:theorems} make these objectives explicit and separate the effects of layer fusion from those of the transient graph construction. Optimizing a specified objective does not, by itself, establish that the resulting communities are scientifically meaningful for a particular application.

    The analytical results also show that the three formulations can respond differently to layer weighting. For {\sf SF}, an optimal layer-weight vector always exists at a vertex of the simplex, although nonvertex optima may also occur when multiple vertices are tied. For {\sf EF}, the additional non-negative heterogeneity term $\gamma\Delta$ and concavity in $\alpha$ for a fixed partition allow a nontrivial mixture of layers to improve the objective, as illustrated by the counterexample in Remark~\ref{remark:ef_interior_counterexample}. The {\sf LF} objective is conditional on the fixed construction operator $\tau$, since {\sf LF} first modifies the layers according to the independently obtained layer-wise partitions and then applies {\sf EF} to the resulting transformed multiplex network. The exact toy experiments illustrate these objective-level properties, while the heuristic experiments additionally show how the choice of optimizer can affect the observed results.

    \noindent\textbf{Transport networks.} Among the historical network experiments, London Stations is the only dataset whose highest median EF score is recorded within the simplex interior. The winning mixture consists of Tube, Overground, and DLR layers, while SF chooses the Tube layer only. The transport modes connect stations in different ways, which may make their combination useful under EF. Through the variable $\Delta$, the algebra considers the effect of heterogeneous community-wise degree volumes. However, it does not provide a causal explanation for the observed partition. On Major Airlines, EF and SF select the British Airways layer; LF selects Lufthansa. Here, no tested interior mixture has a higher median EF score than the best single-layer EF configuration. With normalized layers, we interpret these preferences in terms of the tested modularity objective rather than the number of flights or edges.

    \noindent\textbf{Node-attributed networks.} For Political Blogs, both EF and SF favor the attribute layer. Only same-leaning hyperlinks are retained in this layer, so it separates political leanings by construction. The result is a useful consistency check, but it does not independently validate the inferred communities. In the historical Cora experiment, which used subject-label features, the citation layer gives the highest observed EF median. Adding that sparsified attribute layer did not improve the best observed EF median, despite a nonzero heterogeneity correction. That experiment cannot independently validate recovery of subject categories.

    \noindent\textbf{Label-free external validation.} Section~\ref{sec:cora_external} supplies a separate evaluation in which subject categories are used only after optimization. Equal-weight fusion yields higher median subject agreement, measured by both ARI and NMI, than either single-layer baseline for both optimizers and all three methods, even though their own objectives select citations alone. The external criterion therefore supports using word information in this case, whereas the internal optimum does not. This comparison concerns one fixed network and the tested graph construction, weight grid, and optimization settings. Dispersion across seeds describes optimization variability, not variation across networks. The result neither establishes equal weights as generally optimal nor makes subject categories definitive community ground truth.

    \noindent\textbf{Clinical-trial networks.} In the two-layer Immune Trials network, EF and SF favor intervention similarity over the text layer. In the three-layer network they favor shared sponsors or collaborators over intervention and condition similarity. LF favors the intervention layer in both cases, illustrating how the transient construction can change the preferred layer under its own objective. The largest three-layer $\Delta$ and $\Delta'$ values occur on the condition--sponsor edge of the simplex, but their exploratory winner-gap intervals include zero. The auxiliary maxima therefore do not identify a unique preferred mixture. Nor do they replace the vertex winners for the main objectives. Although their node sets are the same, these two trial networks cannot be compared to isolate the effect of adding a third layer, because their layer content differs.

    \noindent\textbf{Interpretation and limits.} In the historical experiments, a source method need not return the best heuristic partition even under its own objective: for example, the Leiden SF partitions give the best observed EF score on Cora, while LF partitions give the best observed EF score on the two-layer Immune Trials network. Thus the tables compare target objectives and returned partitions separately. For these two Leiden comparisons, the exploratory winner-gap intervals include zero, so they do not establish a stable ranking of source methods. Higher EF values can reflect the heterogeneity bonus, and absolute LF values use a modified graph; neither is evidence of universal scientific superiority. Modularity and external community labels answer different questions \cite{Hric2014,Peel2017,PeelTheory2022}. The analytical results are exact under the assumptions of shared nodes, no interlayer edges, non-negative layer weights that are normalized, and a common resolution parameter. LF also requires the stated fixed construction. The empirical conclusions can change with the preprocessing, optimizer, seeds, and grid resolution. The bootstrap intervals describe seed variability for a given configuration; they do not establish that the weights are globally optimal. London's interior best median shows that a mixture beats the tested single-layer configurations under the reported EF criterion. It does not establish concavity of the empirical profile: Theorem~\ref{theorem:ef_fixed_partition} keeps the partition $C$ fixed as $\alpha$ varies, whereas the partition returned by a heuristic can change between grid points. Even the exact profile $\max_C Q_{\sf EF}(G,C;\gamma,\alpha)$ is a pointwise maximum of concave functions and need not itself be concave.

    The comparison identifies what these representative fusion methods optimize and when layer fusion can change objective values or preferred weights. These analytical statements are separate from empirical optimizer rankings. Evaluating returned partitions under other objectives may reveal differences hidden by a method's own score. The comparisons depend on the standing assumptions. To choose scientifically meaningful communities, we still need a criterion suited to the application. This analysis could be expanded to incorporate layer-specific resolution parameters, additional formulations of multiplex modularity \cite{GadarAbonyi2024}, richer benchmarks such as mABCD \cite{Krainski2025}, and additional LF transient constructions.

\section*{Funding}

This research received no external funding.

\section*{Acknowledgements}

We are grateful to Timofey Gradov and Klavdiya Bochenina for valuable comments on an earlier version of this
manuscript.

\section*{Competing interests}

The authors declare no competing interests.

\section*{Data and code availability}

Source code, software requirements, numerical checks of the EF--SF and LF identities, and materials for reproducing the historical analysis are available in the \href{https://gitlab.com/aipetr/modularity_objective_functions}{public repository}. The separate label-free Cora validation is documented in the accompanying electronic supplement, including its fixed analysis contract, executable code snapshot, input checksums, partitions, and per-run scores. Its results are distinguished from the historical Cora experiment, which included subject-label features.

Clinical-trial data sources and reuse terms are given in Appendix~\ref{app:historical}.

\FloatBarrier
\bibliographystyle{unsrtnat}
\bibliography{literature}

\clearpage
\begin{appendices}

\section{Proofs of the objective identities and concavity}
\label{app:proofs}
\subsection{Proof of Theorem~\ref{theorem_ef_connected_to_sf}}
\label{app:proof_ef}
\begin{proof}
        Fix $C$. We first rewrite the ingredients of  (\ref{eq:modularity}) in terms of (\ref{eq:weights_normalization}) for $\hat{G}_{\alpha}$:
        \begin{equation}
            \label{ingredients}
            \begin{split}
            &\hat{w}(e_{ij})=\sum_{l=1}^L \alpha_l w_l(e_{ij}), \quad \sum_{l=1}^L \alpha_l = 1,\\
            &\hat{\kappa}(v_i)=\sum_{l=1}^L \alpha_l \kappa_l(v_i), \quad \kappa_l(v_i)=\sum_{v_j \in \mathcal{V}} w_l(e_{ij}), \\
            &\hat{m}=\sum_{l=1}^L \alpha_l m_l = 1, \quad m_l = \sum_{e_{ij}\in\mathcal{E}} w_l(e_{ij})=\frac{1}{2} \sum_{v_i, v_j \in\mathcal{V}} w_l(e_{ij})=1.
            \end{split}
        \end{equation}
        
        Furthermore, for any network $g$ and partition $C$:
        \begin{equation}
            \label{kk}
            \begin{split}
             \sum_{v_i, v_j \in C_k} &\kappa(v_i) \kappa(v_j)=  \sum_{v_i \in C_k} \kappa(v_i) \sum_{v_j \in C_k} \kappa(v_j)  = \left[ \sum_{v_i \in C_k} \kappa(v_i) \right] ^2.
            \end{split}
        \end{equation}
        
        The expressions in Eqs.~(\ref{ingredients}) and (\ref{kk}) give for the right-hand side of Eq.~ (\ref{eq:preliminary_Q_EF}):
        \begin{equation}
            \label{early_sim_modularity}
            \begin{split}
            Q_{\sf EF}&(G,C; \gamma,\alpha) = \\
            &=\frac{1}{2} \sum_{k=1}^K \sum_{v_i, v_j \in C_k} \left[ \hat{w}(e_{ij}) - \gamma\frac{\hat{\kappa}(v_i) \hat{\kappa}(v_j)}{2}  \right]\\
            &=\frac{1}{2} \sum_{k=1}^K  \sum_{v_i, v_j \in C_k} \sum_{l=1}^L \alpha_l w_l(e_{ij}) -\frac{1}{2} \sum_{k=1}^K \frac{\gamma}{2} \left[ \sum_{v_i \in C_k} \sum_{l=1}^L \alpha_l \kappa_l(v_i) \right]^2 \\
            &=\sum_{l=1}^L\frac{\alpha_l}{2} \sum_{k=1}^K \sum_{v_i, v_j \in C_k}  \left[ w_l(e_{ij}) - \gamma\frac{\kappa_l(v_i) \kappa_l(v_j)}{2}  \right]\\
            &\qquad + \sum_{l=1}^L\frac{\alpha_l}{2} \sum_{k=1}^K \sum_{v_i, v_j \in C_k}  \left[ \gamma\frac{\kappa_l(v_i) \kappa_l(v_j)}{2}  \right] \\
            &\qquad - \frac{1}{2} \sum_{k=1}^K \frac{\gamma}{2} \left[ \sum_{v_i \in C_k} \sum_{l=1}^L \alpha_l \kappa_l(v_i) \right]^2
            \end{split}
        \end{equation}

        In terms of Eq.~(\ref{eq:modularity}), we get
        \begin{equation}
            \begin{split}
            Q_{\sf EF}&(G,C; \gamma,\alpha)= \\
            &=\sum_{l=1}^L \alpha_l Q(g_l,C; \gamma) \\
            &\qquad + \frac{1}{4}\gamma \sum_{k=1}^K \Bigg[\sum_{l=1}^L \alpha_l \bigg[ \sum_{v_i \in C_k} \kappa_l(v_i) \bigg]^2  - \bigg[ \sum_{l=1}^L \alpha_l  \sum_{v_i \in C_k} \kappa_l(v_i) \bigg]^2 \Bigg].
            \end{split}
        \end{equation}
        
        This is the identity in Eq.~(\ref{composite_modularity_dependency}). Since $\sum_{l=1}^L\alpha_l=1$ and $x\mapsto x^2$ is convex, \textit{Jensen's inequality} \cite{Hardy1952} gives
        \begin{equation}
            \label{jensen-modularity}
            \bigg[ \sum_{l=1}^L \alpha_l  \sum_{v_i \in C_k} \kappa_l(v_i) \bigg]^2\le \sum_{l=1}^L \alpha_l \bigg[ \sum_{v_i \in C_k} \kappa_l(v_i) \bigg]^2.
        \end{equation}
        
        Thus, the inequality in Eq.~(\ref{composite_modularity_diff_pos}) holds true.
        
        A weighted variance is zero exactly when its values are constant on $\operatorname{supp}(\alpha)$. This is immediate for singleton support and follows from the strict convexity of $x\mapsto x^2$ otherwise. Hence $\Delta(G,C;\alpha)=0$ exactly when Eq.~(\ref{condition_jensen_community}) holds for every $k$. The proof also applies to normalized layers with self-loops under the stated adjacency convention: it uses symmetry, degree sums, and total weight, but does not require a zero diagonal. Thus, the theorem applies to the transient multiplex $G'$ used below.
    \end{proof}

\subsection{Proof of Lemma~\ref{lemma_lf_connected_to_sf}}
\label{app:proof_lf}
\begin{proof}
        We first rewrite modularity in Eq.~(\ref{eq:modularity}) for the transient graph $g'$ under normalization in Eq.~(\ref{eq:weights_normalization}):
        \begin{equation}
            \label{th3_1}
            \begin{split}
            Q(g',C;\gamma)&= \frac{1}{2} \sum_{k=1}^K  \sum_{v_i, v_j \in C_k} \left[ w'(e_{ij}) - \gamma\frac{\kappa'(v_i) \kappa'(v_j)}{2}  \right] \\
            &= \frac{1}{2} \sum_{k=1}^K  \sum_{v_i, v_j \in C_k} \left[ w(e_{ij}) - \gamma\frac{\kappa(v_i) \kappa(v_j)}{2}  \right]\\
            &\quad + \frac{1}{2} \sum_{k=1}^K  \sum_{v_i, v_j \in C_k} \left( w'(e_{ij}) - w(e_{ij}) \right) \\
            &\quad - \frac{1}{4} \gamma \sum_{k=1}^K  \sum_{v_i, v_j \in C_k} \left(  \kappa'(v_i) \kappa'(v_j) - \kappa(v_i) \kappa(v_j) \right).
            \end{split}
        \end{equation}
        
        Notice that the first term in the latter expression is $Q(g,C;\gamma)$. Using the second term and applying (\ref{kk}) to the third term gives the exact definition of $\theta(g,g',C; \gamma)$ given in Eq.~(\ref{omega}).

        The inequality in Eq.~(\ref{omega-transient}) follows immediately from Eq.~(\ref{omega}) and the {\sf LF} method's properties in Eqs.~(\ref{LF_construction_1}) and (\ref{LF_construction_2}) for a single-layer network. Indeed,
        \begin{equation*}
            \label{omega_tilde_C}
            \begin{split}
            \theta(g,g',C'; \gamma) =
            & \sum_{k=1}^{K'} p_k,\qquad p_k:= \frac{1}{2}\sum_{v_i, v_j \in C'_k} \left( w'({e}_{ij}) - w({e}_{ij}) \right) \ge 0.
            \end{split}
        \end{equation*}
        
        This also implies Eq.~(\ref{weight_estimation_zero}), i.e. that $\theta(g,g',C'; \gamma)\equiv 0$ if and only if $p_k\equiv 0$ for all $k=1,\ldots,K'$.  
    \end{proof}

\subsection{Proof of Theorem~\ref{theorem:ef_fixed_partition}}
\label{app:proof_concavity}
\begin{proof}
        For $\gamma=0$, $Q_{\sf EF}$ is affine in $\alpha$ directly from Eqs.~(\ref{eq:modularity}) and (\ref{eq:preliminary_Q_EF}), and the two-layer formula follows without a quadratic term. Assume $\gamma>0$ for the remaining argument.

        Theorem~\ref{theorem_ef_connected_to_sf} yields
        \begin{equation*}
            Q_{\sf EF}(G,C;\gamma,\alpha)=\sum_{l=1}^L \alpha_l Q(g_l,C;\gamma)+\gamma\Delta(G,C;\alpha).
        \end{equation*}
        The first term is affine in $\alpha$. For each community $C_k$, the contribution to $\Delta$ has the form
        \begin{equation*}
            \frac{1}{4}\left(\sum_{l=1}^L \alpha_l d_{lk}(C)^2-\left[\sum_{l=1}^L \alpha_l d_{lk}(C)\right]^2\right),
        \end{equation*}
        which is affine minus a convex quadratic term and is therefore concave on $\mathcal{S}$. Summing over $k$ preserves concavity.

        For $L=2$, the identity
        \begin{equation*}
            ad_{1k}^2+(1-a)d_{2k}^2-(ad_{1k}+(1-a)d_{2k})^2=a(1-a)(d_{1k}-d_{2k})^2
        \end{equation*}
        gives Eq.~(\ref{eq:ef_two_layer_alpha}).
    \end{proof}

\section{Additional exact toy results}
\label{app:toy_detail}
The following figures supplement the discussion of the second toy network in Section~\ref{exp_toy}.

\begin{figure}[!htbp]
\centering\begin{minipage}{4.25in}
        \centering
    \scalebox{1.0}{%
    \begin{minipage}{\textwidth}
    \centering
        \begin{subfigure}[b]{0.49\textwidth}
            \centering\includegraphics[width=\linewidth]{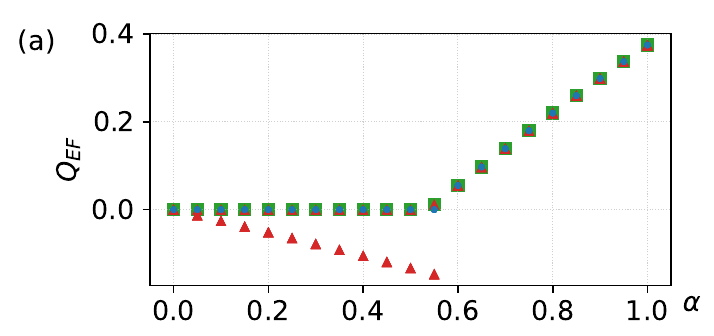}
            \phantomsubcaption\label{fig:toy2_ef}
        \end{subfigure}
        \begin{subfigure}[b]{0.49\textwidth}
            \centering\includegraphics[width=\linewidth]{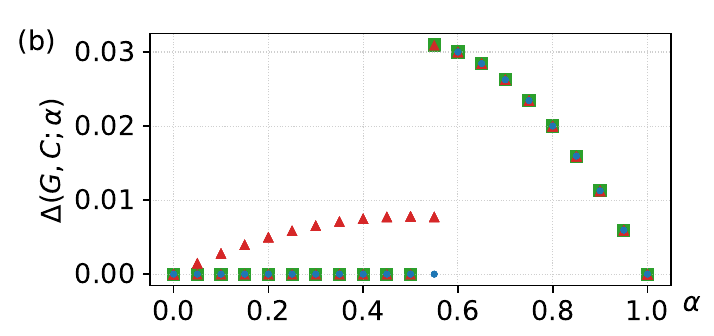}
            \phantomsubcaption\label{fig:toy2_delta}
        \end{subfigure}
        \begin{subfigure}[b]{0.49\textwidth}
            \centering\includegraphics[width=\linewidth]{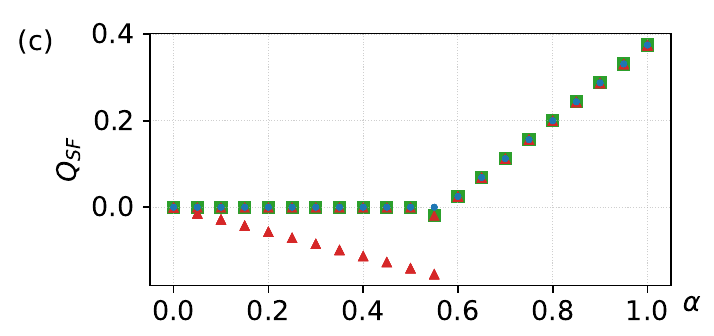}
            \phantomsubcaption\label{fig:toy2_sf}
        \end{subfigure}
        \begin{subfigure}[b]{0.49\textwidth}
            \centering\includegraphics[width=\linewidth]{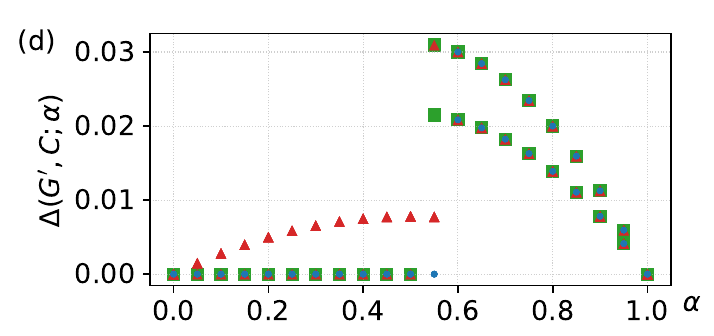}
            \phantomsubcaption\label{fig:toy2_delta_new}
        \end{subfigure}
        \begin{subfigure}[b]{0.49\textwidth}
            \centering\includegraphics[width=\linewidth]{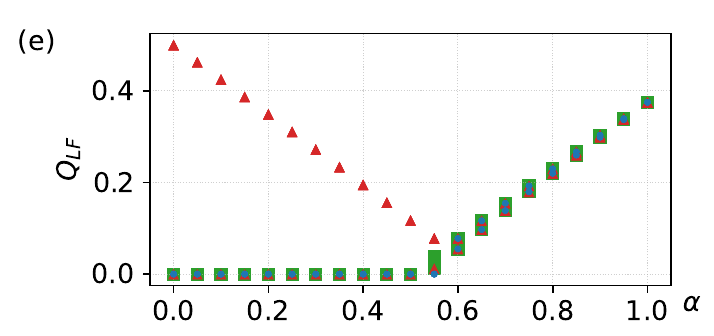}
            \phantomsubcaption\label{fig:toy2_lf}
        \end{subfigure}
        \begin{subfigure}[b]{0.49\textwidth}
            \centering\includegraphics[width=\linewidth]{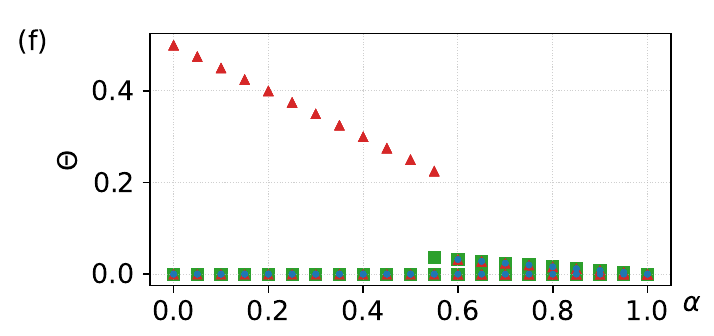}
            \phantomsubcaption\label{fig:toy2_theta}
        \end{subfigure}
        
    \toymethodlegend
    \end{minipage}%
    }
    \caption{Toy Example 2. Brute force results of the {\sf EF}, {\sf SF} and {\sf LF} methods (producing correspondingly the partitions $C_{\sf EF}$, $C_{\sf SF}$ and $C_{\sf LF}$) in terms of the functions $Q_{\sf EF}$, $Q_{\sf SF}$, $Q_{\sf LF}$, $\Delta(G,C;\alpha)$, $\Delta(G'^{\tau},C;\alpha)$ and $\Theta$ for $\alpha=(a,1-a)$ with $a\in\{0.00,0.05,\ldots,1.00\}$. As in Figure~\ref{brute_force_example_result_1}, the $Q_{\sf LF}$ axis is evaluated on the transient layers $G'^{\tau}$ and is therefore not directly comparable to the $Q_{\sf EF}$/$Q_{\sf SF}$ axes in the other objective panels.}
        \label{brute_force_example_result_2}
    \end{minipage}
\end{figure}

\begin{figure}[t]
\centering\begin{minipage}{4.25in}
        \centering
        \includegraphics[width=0.765\linewidth]{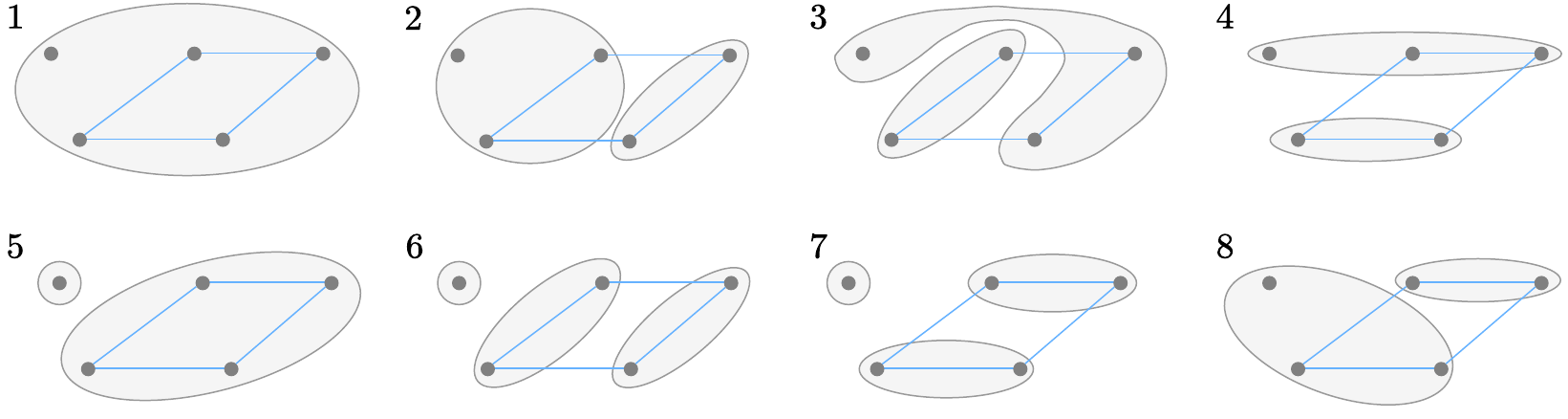}
        \\
        \includegraphics[width=0.05\linewidth]{figures/blue_line.pdf}
        link with weight $\frac{1}{4}$
        \caption{Transient partition diversity for the lower layer of the network in Figure~\ref{brute-force-examples}b. The maximum modularity value in this case is $0$.}
        \label{fig:transient_diversity}
    \end{minipage}
\end{figure}

\section{Historical network experiments: data and detailed results}
\label{app:historical}
\subsection{Real-world two-layer networks and modularity optimization by Louvain and Leiden}
    \begin{table}[b]
        \caption{Statistics of the analyzed networks. Link counts follow layer order (layer 1 / layer 2 / layer 3, where applicable); node-attributed datasets are shown after their two-layer transformation.}
        \label{tab1}
        \centering\small\setlength{\tabcolsep}{2pt}
        \begin{tabular}{|lrrrr|}
\hline
Dataset & Nodes & Layers & Links per layer & Attr. dim. \\
\hline
Political Blogs & 1224 & 2 & 16715/15140 & 1 \\
Cora & 2708 & 2 & 5278/49314 & 1440 \\
Major Airlines & 128 & 3 & 69/244/66 & --- \\
London Stations & 369 & 3 & 312/83/46 & --- \\
Immune Trials (2L) & 1325 & 2 & 15909/24777 & 3670 \\
Immune Trials (3L) & 1325 & 3 & 15909/30537/14393 & --- \\
\hline
\end{tabular}
\par\smallskip\raggedright Political Blogs and Cora are node-attributed networks; the other datasets are multiplex networks. Attribute formats are binary for Political Blogs, binary and categorical for Cora, and real-valued for Immune Trials (2L). Dashes indicate that attribute dimension and format do not apply.

    \end{table}

    We use undirected, cleaned, shared-node versions of the following public datasets. Table~\ref{tab1} reports the processed counts used in the experiments after duplicate removal, self-loop handling, and layer synchronization. The analyzed common node set consists of vertices incident to at least one retained layer edge. A node isolated in one layer but present in another is retained in every layer.
    \begin{itemize}
        \item \href{https://networkdata.ics.uci.edu/data/polblogs/}{\textit{Political Blogs}}~\cite{PoliticalBlogsData} is a network of 1,224 webblogs (nodes) on US politics. In the processed undirected artifact used here, the structural layer contains 16,715 cleaned hyperlinks. It provides a relatively simple node-attributed sanity check after the structure/attribute transformation to a two-layer multiplex, rather than a demanding fusion benchmark. Each node has a single binary attribute describing its political leaning (liberal or conservative). The attribute layer $g_2$ retains exactly those structural links whose endpoints share the leaning label, normalized to a total weight of 1; see~(\ref{polblogs-attribute-layer}). By construction, it contains no cross-leaning edges, so its signal is fully aligned with the leaning labels.
        \item \href{https://linqs-data.soe.ucsc.edu/public/lbc/cora.tgz}{\textit{Cora}}~\cite{CoraData} is a network of machine learning papers with 2,708 papers (nodes). In the processed undirected artifact used here, the structural layer contains 5,278 cleaned citation links. It provides a richer node-attributed case because the attribute layer is derived from high-dimensional textual and categorical data. Each node is attributed with a 1,433-dimensional binary vector indicating the absence/presence of words from the dictionary of words collected from the corpus of papers, and one categorical attribute, for which we have done one-hot encoding and received 7 additional attributes for each vertex. The Cora attribute layer uses the sparse cosine nearest-neighbor construction defined below.
        \item {\it Immune Trials} is a two-layer multiplex network over registered clinical trials. Unlike Political Blogs and Cora, it is \emph{not} obtained by the node-attributed transformation of Section~\ref{sec:multiplex_networks}: registered trials carry no natural link structure, so neither layer is structural in the hyperlink or citation sense. Both layers are attribute-derived similarity graphs on a shared node set, which is all our analysis requires: the results assume only normalized non-negative layers on a common node set, never that one layer is the structure of an underlying graph. The selection rule retains interventional Phase~3 trials carrying at least one of the Medical Subject Headings (MeSH) condition terms \emph{Arthritis, Rheumatoid}, \emph{Asthma}, \emph{Colitis, Ulcerative}, \emph{Crohn Disease} or \emph{Psoriasis}, which yields 1{,}578 trials; the 1{,}325 of these that carry condition, intervention and sponsor annotations together with at least 200 characters of free text form the node set. We take $g_1$ to be the top-25 cosine nearest-neighbor graph over binary shared-intervention MeSH incidence (15{,}909 edges) and $g_2$ to be the top-25 cosine nearest-neighbor graph over a term frequency--inverse document frequency (TF-IDF) encoding of the brief summary and eligibility criteria text (3{,}670 terms, 24{,}777 edges). It provides a two-layer case in which the layers are only weakly aligned: they share $13.5\%$ of their edges by Jaccard overlap, whereas for Political Blogs the attribute layer is by construction a subgraph retaining about $91\%$ of the structural edges.
    \end{itemize}
    
\noindent\textbf{Clinical-trial data source and reuse.} \href{https://aact.ctti-clinicaltrials.org/}{Aggregate Analysis of ClinicalTrials.gov (AACT) Database} \cite{AACT2026}, Clinical Trials Transformation Initiative (CTTI); source: ClinicalTrials.gov, U.S. National Library of Medicine (NLM). ``Courtesy of the National Library of Medicine.'' The source tables have a maximum recorded registry-update date of 2026-06-18; the local cohort was extracted on 2026-08-18. The original snapshot export date is not recorded. These derived networks do not reflect the most current data available from NLM. Users agree to hold NLM and the U.S. Government harmless from any liability resulting from errors in the data; the results reported here are not endorsed by NLM, NIH or CTTI. A minority of free-text registry fields may embed third-party material protected by non-US copyright law; only derived and aggregated features are used, and the raw export is not redistributed.

    For Political Blogs and Cora, the transformation in Section~\ref{sec:multiplex_networks} gives $G=(g_1,g_2)$, where $g_1$ is structural and $g_2$ is attribute-based. For nonzero, non-negative attribute vectors and $i\ne j$, the unnormalized cosine similarity is
    \begin{equation}
        \label{mu-nu}
        s_{ij}=\frac{A(v_i)\cdot A(v_j)}{\|A(v_i)\|_2\|A(v_j)\|_2}.
    \end{equation}
    Here $s_{ij}\in[0,1]$ by non-negativity and the Cauchy--Schwarz inequality. If either vector is zero, we set $s_{ij}=0$ by convention, as in the implementation; we also set $s_{ii}=0$ to exclude self-loops. These similarities are sparsified and normalized to obtain layer weights.

    For Political Blogs, the attribute layer is the same-label subgraph of the structural layer, induced by the binary leaning attribute:
    \begin{equation}
        \label{polblogs-attribute-layer}
        \begin{split}
        w_2(e_{ij}) &= \frac{\mathbf{1}\{w_1(e_{ij})>0,\ \ell(v_i)=\ell(v_j)\}}{Z},\\
        Z &= \sum_{i<j}\mathbf{1}\{w_1(e_{ij})>0,\ \ell(v_i)=\ell(v_j)\},
        \end{split}
    \end{equation}
    where $\ell(v_i)\in\{\text{liberal},\text{conservative}\}$ and $w_1$ denotes the structural-layer weights. By construction, every edge of $g_2$ is a structural hyperlink between two same-leaning blogs and $\sum_{e_{ij}\in\mathcal{E}} w_2(e_{ij})=1$. This is the convention used in the public artifact, where the cleaned dataset yields $|\mathcal{V}|=1{,}224$ and an attribute layer retaining $15{,}140$ of the $16{,}715$ structural edges (about $91\%$ of hyperlinks connect same-leaning blogs). The construction reflects that the leaning attribute is single and binary: the cosine similarity of one-hot leaning encodings is exactly the same-label indicator, and the artifact evaluates this indicator on the structural edge set, so~(\ref{polblogs-attribute-layer}) is the direct description of the layer actually used.

    For Cora, the canonical processed layer is
    \begin{equation}
        W_A=\mathrm{Normalize}\bigl(\mathrm{Sym}_{\mathrm{union}}(\mathrm{kNN}_{25}((s_{ij})_{i,j=1}^n))\bigr),
    \end{equation}
    where $\mathrm{kNN}_{25}$ selects up to 25 neighbors with the largest positive cosine similarities for each node, excluding the node itself. The union symmetrization $\mathrm{Sym}_{\mathrm{union}}$ retains an undirected edge if either endpoint selects the other, with its cosine weight. Finally, $\mathrm{Normalize}$ divides each retained edge weight by the sum of all retained edge weights, as in (\ref{eq:weights_normalization}).

    For Immune Trials, both layers use the same rule as the Cora attribute layer, differing only in the feature matrix that enters the cosine similarity: binary shared-intervention MeSH incidence for $g_1$ and the TF-IDF text encoding for $g_2$. One caveat is specific to the binary-incidence layers. Binary features produce many exactly tied cosine similarities, so the choice of which equally similar neighbors occupy the last of the $k=25$ slots is not determined by the similarity values alone: a tie straddles the $k$-th boundary for $714$ of $1{,}325$ nodes in the intervention layer and for $1{,}320$ of $1{,}325$ in the condition layer used below, against only $8$ of $1{,}325$ in the continuous TF-IDF layer. The layers materialized in the artifact are therefore the canonical ones, and edge counts recomputed by a different implementation of the same construction may differ by a few tenths of a percent without changing component structure.

    For each two-layer comparison, all three fusion methods use the same optimizer family. EF and LF follow Sections~\ref{EF_description} and \ref{LF_description}; LF uses that optimizer for both transient and final partitions. For SF, the Louvain runs use our multiplex extension described next, and the Leiden runs use the multiplex extension of \cite{Traag2019}.

Our Louvain extension retains the two phases of \cite{Blondel2008}. First, it makes node moves that improve the convex combination of layer modularities in (\ref{eq:SF_objective_function_intro}), using $\alpha$ as an input. Second, it aggregates the communities into nodes in each of the $L$ layers. These phases repeat until no improving move remains.

\begin{figure}[!htbp]
\centering\begin{minipage}{4.25in}
    \centering
    \scalebox{1.0}{%
    \begin{minipage}{\textwidth}
    \centering
    \begin{subfigure}[b]{0.49\textwidth}
        \centering\includegraphics[width=\linewidth]{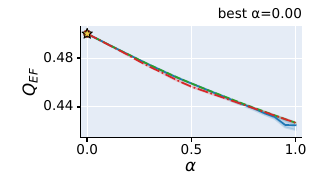}
    \end{subfigure}
    \hfill
    \begin{subfigure}[b]{0.49\textwidth}
        \centering\includegraphics[width=\linewidth]{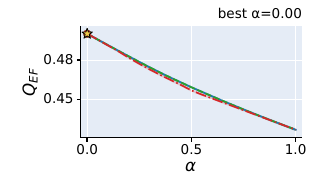}
    \end{subfigure}
    \hfill
    \begin{subfigure}[b]{0.49\textwidth}
        \centering\includegraphics[width=\linewidth]{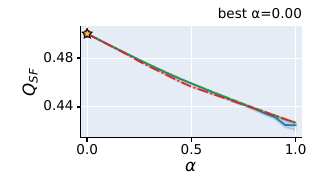}
    \end{subfigure}
    \hfill
    \begin{subfigure}[b]{0.49\textwidth}
        \centering\includegraphics[width=\linewidth]{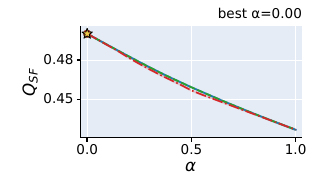}
    \end{subfigure}
    \hfill
    \begin{subfigure}[b]{0.49\textwidth}
        \centering\includegraphics[width=\linewidth]{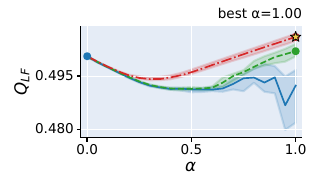}
    \end{subfigure}
    \hfill
    \begin{subfigure}[b]{0.49\textwidth}
        \centering\includegraphics[width=\linewidth]{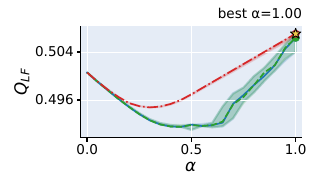}
    \end{subfigure}
    
    \partitionlegend

    \end{minipage}%
    }
    \caption{Political Blogs graph modularity optimization performed by Louvain-based (left) and Leiden (right) algorithms. Target functions. The $Q_{\sf LF}^{\tau}$ row is evaluated on the transient layers $G'^{\tau}$ and is not on the same raw scale as $Q_{\sf EF}$/$Q_{\sf SF}$; compare $Q_{\sf LF}$ values across $\alpha$ within each panel, not against the $Q_{\sf EF}$/$Q_{\sf SF}$ panels.}
    \label{fig:PolBlogs_target_functions}
\end{minipage}
\end{figure}

\begin{figure}[!htbp]
\centering\begin{minipage}{4.25in}
    \centering
    \scalebox{1.0}{%
    \begin{minipage}{\textwidth}
    \centering
    \begin{subfigure}[b]{0.49\textwidth}
        \centering\includegraphics[width=\linewidth]{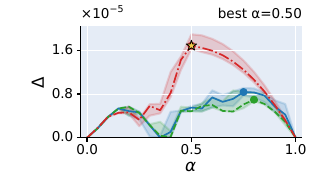}
    \end{subfigure}
    \hfill
    \begin{subfigure}[b]{0.49\textwidth}
        \centering\includegraphics[width=\linewidth]{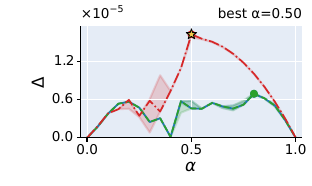}
    \end{subfigure}
    \hfill
    \begin{subfigure}[b]{0.49\textwidth}
        \centering\includegraphics[width=\linewidth]{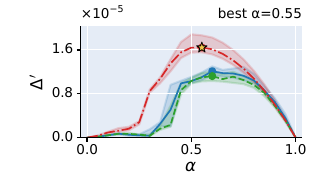}
    \end{subfigure}
    \hfill
    \begin{subfigure}[b]{0.49\textwidth}
        \centering\includegraphics[width=\linewidth]{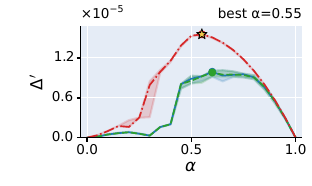}
    \end{subfigure}
    \hfill
    \begin{subfigure}[b]{0.49\textwidth}
        \centering\includegraphics[width=\linewidth]{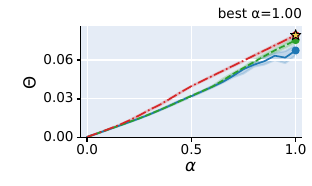}
    \end{subfigure}
    \hfill
    \begin{subfigure}[b]{0.49\textwidth}
        \centering\includegraphics[width=\linewidth]{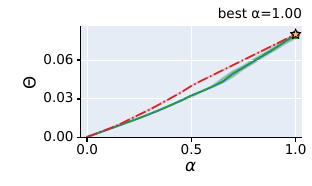}
    \end{subfigure}
    
    \partitionlegend

    \end{minipage}%
    }
    \caption{Political Blogs graph modularity optimization performed by Louvain-based (left) and Leiden-based (right) algorithms. Correction terms and auxiliary quantities.}
    \label{fig:PolBlogs_particles}
\end{minipage}
\end{figure}

\begin{figure}[!htbp]
\centering\begin{minipage}{4.25in}
    \centering
    \scalebox{1.0}{%
    \begin{minipage}{\textwidth}
    \centering
    \begin{subfigure}[b]{0.49\textwidth}
        \centering\includegraphics[width=\linewidth]{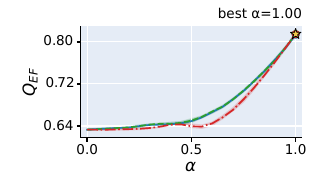}
    \end{subfigure}
    \hfill
    \begin{subfigure}[b]{0.49\textwidth}
        \centering\includegraphics[width=\linewidth]{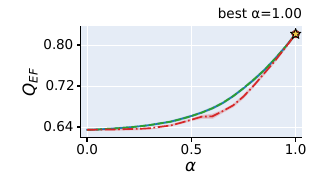}
    \end{subfigure}
    \hfill
    \begin{subfigure}[b]{0.49\textwidth}
        \centering\includegraphics[width=\linewidth]{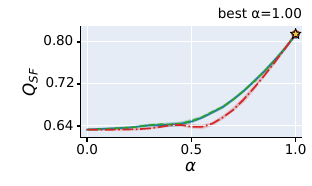}
    \end{subfigure}
    \hfill
    \begin{subfigure}[b]{0.49\textwidth}
        \centering\includegraphics[width=\linewidth]{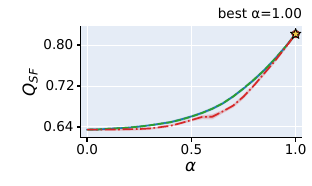}
    \end{subfigure}
    \hfill
    \begin{subfigure}[b]{0.49\textwidth}
        \centering\includegraphics[width=\linewidth]{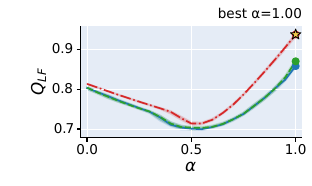}
    \end{subfigure}
    \hfill
    \begin{subfigure}[b]{0.49\textwidth}
        \centering\includegraphics[width=\linewidth]{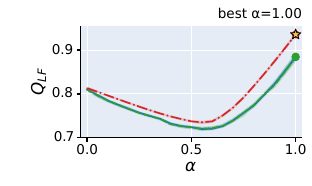}
    \end{subfigure}
    
    \partitionlegend

    \end{minipage}%
    }
    \caption{Cora graph modularity optimization performed by Louvain-based (left) and Leiden-based (right) algorithms. Target functions. The $Q_{\sf LF}^{\tau}$ row is evaluated on the transient layers $G'^{\tau}$ and is not on the same raw scale as $Q_{\sf EF}$/$Q_{\sf SF}$.}
    \label{fig:Cora_target_functions}
\end{minipage}
\end{figure}

\begin{figure}[!htbp]
\centering\begin{minipage}{4.25in}
    \centering
    \scalebox{1.0}{%
    \begin{minipage}{\textwidth}
    \centering
    \begin{subfigure}[b]{0.49\textwidth}
        \centering\includegraphics[width=\linewidth]{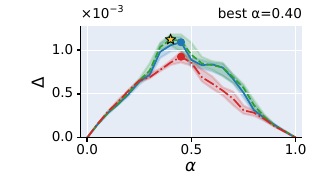}
    \end{subfigure}
    \hfill
    \begin{subfigure}[b]{0.49\textwidth}
        \centering\includegraphics[width=\linewidth]{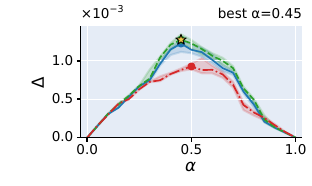}
    \end{subfigure}
    \hfill
    \begin{subfigure}[b]{0.49\textwidth}
        \centering\includegraphics[width=\linewidth]{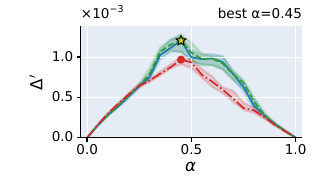}
    \end{subfigure}
    \hfill
    \begin{subfigure}[b]{0.49\textwidth}
        \centering\includegraphics[width=\linewidth]{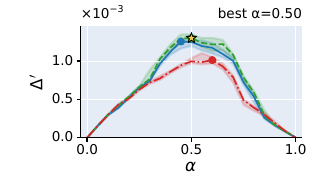}
    \end{subfigure}
    \hfill
    \begin{subfigure}[b]{0.49\textwidth}
        \centering\includegraphics[width=\linewidth]{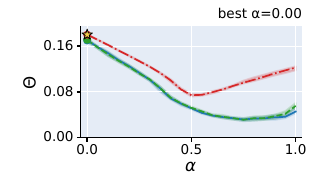}
    \end{subfigure}
    \hfill
    \begin{subfigure}[b]{0.49\textwidth}
        \centering\includegraphics[width=\linewidth]{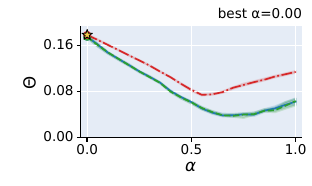}
    \end{subfigure}
    
    \partitionlegend

    \end{minipage}%
    }
    \caption{Cora graph modularity optimization performed by Louvain-based (left) and Leiden-based (right) algorithms. Correction terms and auxiliary quantities.}
    \label{fig:Cora_particles}
\end{minipage}
\end{figure}

\begin{figure}[!htbp]
\centering\begin{minipage}{4.25in}
    \centering
    \scalebox{1.0}{%
    \begin{minipage}{\textwidth}
    \centering
    \begin{subfigure}[b]{0.49\textwidth}
        \centering\includegraphics[width=\linewidth]{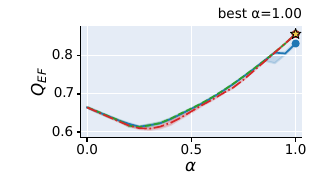}
    \end{subfigure}
    \hfill
    \begin{subfigure}[b]{0.49\textwidth}
        \centering\includegraphics[width=\linewidth]{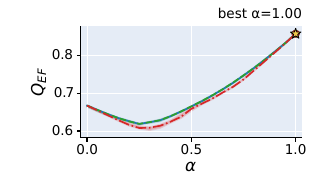}
    \end{subfigure}
    \hfill
    \begin{subfigure}[b]{0.49\textwidth}
        \centering\includegraphics[width=\linewidth]{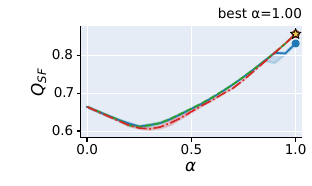}
    \end{subfigure}
    \hfill
    \begin{subfigure}[b]{0.49\textwidth}
        \centering\includegraphics[width=\linewidth]{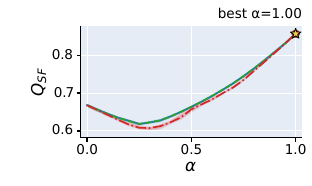}
    \end{subfigure}
    \hfill
    \begin{subfigure}[b]{0.49\textwidth}
        \centering\includegraphics[width=\linewidth]{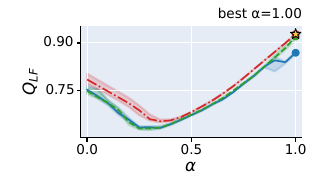}
    \end{subfigure}
    \hfill
    \begin{subfigure}[b]{0.49\textwidth}
        \centering\includegraphics[width=\linewidth]{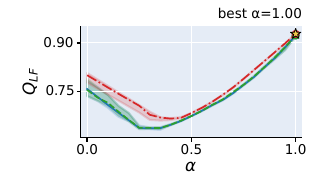}
    \end{subfigure}
    
    \partitionlegend

    \end{minipage}%
    }
    \caption{Immune Trials two-layer multiplex modularity optimization performed by Louvain-based (left) and Leiden-based (right) algorithms. Target functions. The $Q_{\sf LF}^{\tau}$ row is evaluated on the transient layers $G'^{\tau}$ and is not on the same raw scale as $Q_{\sf EF}$/$Q_{\sf SF}$.}
    \label{fig:ctgov_attr_target_functions}
\end{minipage}
\end{figure}

\begin{figure}[!htbp]
\centering\begin{minipage}{4.25in}
    \centering
    \scalebox{1.0}{%
    \begin{minipage}{\textwidth}
    \centering
    \begin{subfigure}[b]{0.49\textwidth}
        \centering\includegraphics[width=\linewidth]{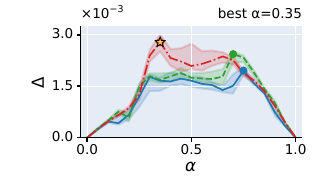}
    \end{subfigure}
    \hfill
    \begin{subfigure}[b]{0.49\textwidth}
        \centering\includegraphics[width=\linewidth]{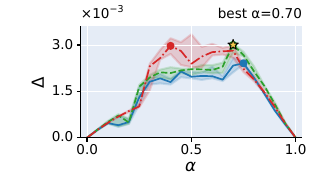}
    \end{subfigure}
    \hfill
    \begin{subfigure}[b]{0.49\textwidth}
        \centering\includegraphics[width=\linewidth]{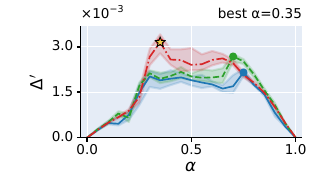}
    \end{subfigure}
    \hfill
    \begin{subfigure}[b]{0.49\textwidth}
        \centering\includegraphics[width=\linewidth]{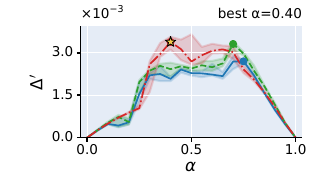}
    \end{subfigure}
    \hfill
    \begin{subfigure}[b]{0.49\textwidth}
        \centering\includegraphics[width=\linewidth]{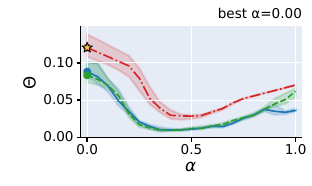}
    \end{subfigure}
    \hfill
    \begin{subfigure}[b]{0.49\textwidth}
        \centering\includegraphics[width=\linewidth]{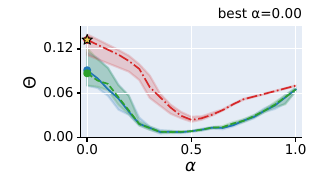}
    \end{subfigure}
    
    \partitionlegend

    \end{minipage}%
    }
    \caption{Immune Trials two-layer multiplex modularity optimization performed by Louvain-based (left) and Leiden-based (right) algorithms. Correction terms and auxiliary quantities.}
    \label{fig:ctgov_attr_particles}
\end{minipage}
\end{figure}

\subsection{Real-world multiplex networks and modularity optimization by Leiden}

The implementations of the {\sf EF} and {\sf LF} methods follow the schemes described in Sections~\ref{EF_description} and \ref{LF_description}, respectively, where the modularity optimization is performed by Leiden \cite{Traag2019} (both for finding transient and final partitions in the {\sf LF} method). The {\sf SF} method here is the extension of Leiden for multiplex networks proposed by the authors of \cite{Traag2019} with the added control of layer contributions by the set of fusion parameters $\alpha$. The three-layer experiments use Leiden for all three methods. For LF cross-objective evaluation, the score $Q_{\sf LF}$ is computed on the transient graph $G'$ produced for the same dataset, optimizer, and seed and reused across the $\alpha$ grid; hence LF scores are conditional on that transient construction rather than being a single seed-independent deterministic scale across all runs.

For the three-layer multiplex experiments, each layer represents a different relation or attribute-derived similarity on the same node set.

We will consider the following publicly available datasets:
\begin{itemize}
    \item \href{https://manliodedomenico.com/data.php}{\textit{London Stations Multiplex Network}} \cite{Domenico8351,LondonMultiplexData}. It was collected in 2013 from the official website of \href{https://www.tfl.gov.uk/}{Transport for London} and manually cross-checked. Nodes are train stations in London and edges encode existing routes between stations. Underground, Overground and Docklands Light Railway (DLR) stations are considered. The analyzed multiplex contains 369 stations and three layers corresponding to: (a) the aggregation to a single weighted graph of the networks of stations corresponding to each underground line (e.g., District, Circle, etc.); (b) the network of stations connected by Overground; (c) the network of stations connected by DLR. The number of edges on layer (a) is 312; on layer (b) -- 83; on layer (c) -- 46.
\item \href{http://complex.unizar.es/~atnmultiplex/}{\textit{Major Airlines Multiplex Network}} \cite{airplanes_dataset,EuropeanAirlinesData}. Each layer naturally aligns with a distinct airline operator, which makes the dataset a clear test case for layer-specific community structure. The analyzed multiplex network consists of 3 layers each one corresponding to one of the three biggest airlines operating in Europe: Air France, Lufthansa, British Airways. The data are taken from the complete list of airlines operating Instrumental Flight Rules (IFR) flights between European airports on a certain day obtained from EUROCON-TROL and the Complex World Network in the context of the SESAR Work Package E42. The source catalogue lists 450 airports; the analysis uses the 128 airports incident to at least one edge in the three selected-airline layers. The remaining 322 catalogue airports are absent from these layers and excluded from the analyzed node set. The number of edges is 69 for the Air France layer, 244 for Lufthansa, and 66 for British Airways.
\item {\it Immune Trials Multiplex Network} (same source, snapshot and license terms as the two-layer Immune Trials network above) is built over exactly the same 1{,}325-node trial set as that network, but modeled as three layers, each recording a different kind of shared registry annotation between the same trials: (a) shared intervention MeSH terms; (b) shared condition MeSH terms; (c) a shared sponsor or collaborator. Each layer is the same top-25 cosine nearest-neighbor construction applied to the corresponding binary incidence matrix. The number of edges on layer (a) is 15{,}909; on layer (b) -- 30{,}537; on layer (c) -- 14{,}393. The three layers have low pairwise edge Jaccard overlaps: $0.041$ for (a)--(b), $0.161$ for (a)--(c) and $0.027$ for (b)--(c). Layer (a) is edge-identical to the layer $g_1$ of the two-layer Immune Trials network, so results on the two datasets are directly comparable node-for-node. They are not, however, a nested pair: the two-layer network's text layer does not appear among the three multiplex layers, so a difference between the two datasets mixes the change in the number of layers with a change in their content and cannot be read as a controlled two-layer versus three-layer comparison.
\end{itemize}

Because every layer is normalized to total weight one before optimization, raw layer-size differences are intentionally removed. The three-layer experiments therefore ask which normalized layer signal dominates or combines best under the objective, not which raw layer contains the most edges or traffic.

\begin{figure}[!htbp]
\centering\begin{minipage}{4.25in}
    \centering
    \scalebox{1.0}{%
    \begin{minipage}{\textwidth}
    \centering
    \begin{subfigure}[b]{0.98\textwidth}
        \centering\includegraphics[width=\linewidth]{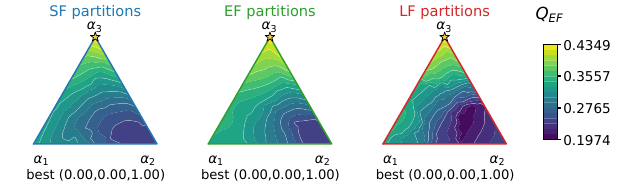}
    \end{subfigure}
    \vspace{0.6em}
    \begin{subfigure}[b]{0.98\textwidth}
        \centering\includegraphics[width=\linewidth]{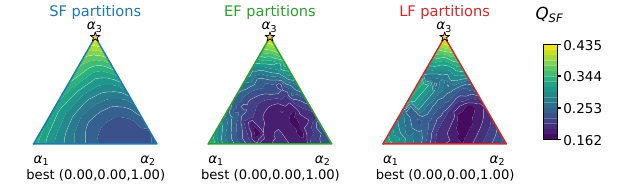}
    \end{subfigure}
    \vspace{0.6em}
    \begin{subfigure}[b]{0.98\textwidth}
        \centering\includegraphics[width=\linewidth]{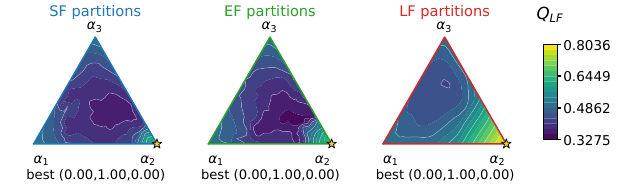}
    \end{subfigure}

    \end{minipage}%
    }
    \caption{Major Airlines multiplex network modularity optimization performed by Leiden-based algorithms. Target functions. The $Q_{\sf LF}^{\tau}$ row is evaluated on the transient layers $G'^{\tau}$ and is not on the same raw scale as $Q_{\sf EF}$/$Q_{\sf SF}$. Layer order is $\alpha_1$: Air France, $\alpha_2$: Lufthansa, $\alpha_3$: British Airways.}
    \label{fig:major_airlines_target}
\end{minipage}
\end{figure}

\begin{figure}[!htbp]
\centering\begin{minipage}{4.25in}
    \centering
    \scalebox{1.0}{%
    \begin{minipage}{\textwidth}
    \centering
    \begin{subfigure}[b]{0.98\textwidth}
        \centering\includegraphics[width=\linewidth]{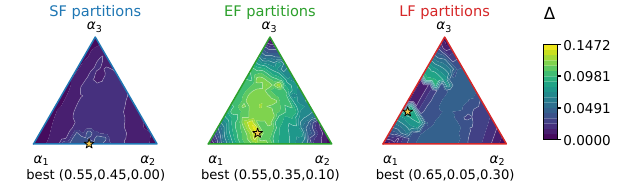}
    \end{subfigure}
    \vspace{0.6em}
    \begin{subfigure}[b]{0.98\textwidth}
        \centering\includegraphics[width=\linewidth]{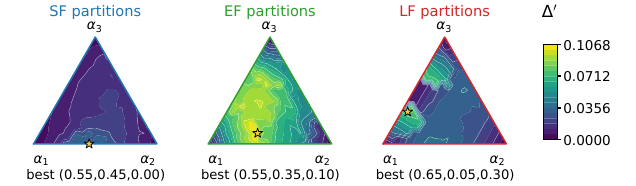}
    \end{subfigure}
    \vspace{0.6em}
    \begin{subfigure}[b]{0.98\textwidth}
        \centering\includegraphics[width=\linewidth]{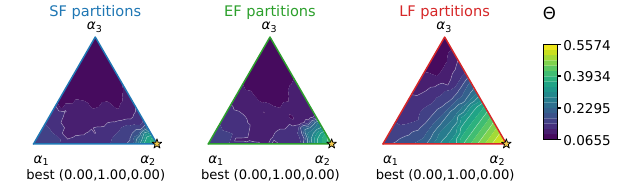}
    \end{subfigure}
    
    \end{minipage}%
    }
    \caption{Major Airlines multiplex network modularity optimization performed by Leiden-based algorithms. Correction terms and auxiliary quantities. Layer order is $\alpha_1$: Air France, $\alpha_2$: Lufthansa, $\alpha_3$: British Airways.}
    \label{fig:major_airlines_aux}
\end{minipage}
\end{figure}

\begin{figure}[!htbp]
\centering\begin{minipage}{4.25in}
    \centering
    \scalebox{1.0}{%
    \begin{minipage}{\textwidth}
    \centering
    \begin{subfigure}[b]{0.98\textwidth}
        \centering\includegraphics[width=\linewidth]{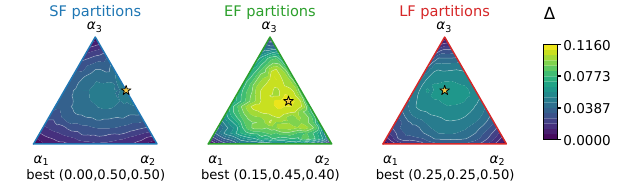}
    \end{subfigure}
    \vspace{0.6em}
    \begin{subfigure}[b]{0.98\textwidth}
        \centering\includegraphics[width=\linewidth]{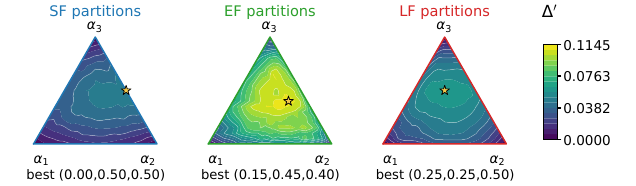}
    \end{subfigure}
    \vspace{0.6em}
    \begin{subfigure}[b]{0.98\textwidth}
        \centering\includegraphics[width=\linewidth]{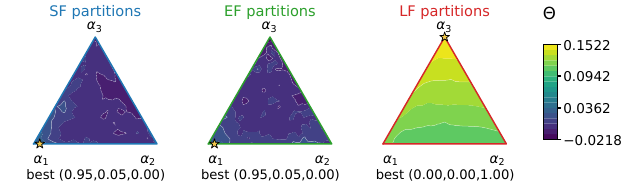}
    \end{subfigure}
    
    \end{minipage}%
    }
    \caption{London Stations multiplex network modularity optimization performed by Leiden-based algorithms. DLR denotes the Docklands Light Railway. Correction terms and auxiliary quantities. Layer order is $\alpha_1$: Tube, $\alpha_2$: Overground, $\alpha_3$: DLR.}
    \label{fig:london_stations_aux}
\end{minipage}
\end{figure}

\begin{figure}[!htbp]
\centering\begin{minipage}{4.25in}
    \centering
    \scalebox{1.0}{%
    \begin{minipage}{\textwidth}
    \centering
    \begin{subfigure}[b]{0.98\textwidth}
        \centering\includegraphics[width=\linewidth]{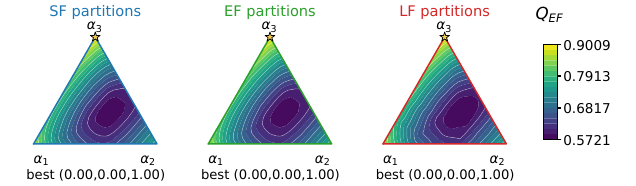}
    \end{subfigure}
    \vspace{0.6em}
    \begin{subfigure}[b]{0.98\textwidth}
        \centering\includegraphics[width=\linewidth]{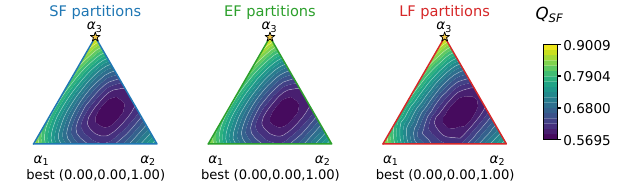}
    \end{subfigure}
    \vspace{0.6em}
    \begin{subfigure}[b]{0.98\textwidth}
        \centering\includegraphics[width=\linewidth]{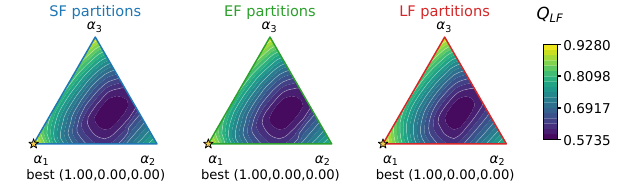}
    \end{subfigure}

    \end{minipage}%
    }
    \caption{Immune Trials multiplex network modularity optimization performed by Leiden-based algorithms. Target functions. The $Q_{\sf LF}^{\tau}$ row is evaluated on the transient layers $G'^{\tau}$ and is not on the same raw scale as $Q_{\sf EF}$/$Q_{\sf SF}$. Layer order is $\alpha_1$: intervention, $\alpha_2$: condition, $\alpha_3$: sponsor.}
    \label{fig:ctgov_multiplex_target}
\end{minipage}
\end{figure}

\begin{figure}[!htbp]
\centering\begin{minipage}{4.25in}
    \centering
    \scalebox{1.0}{%
    \begin{minipage}{\textwidth}
    \centering
    \begin{subfigure}[b]{0.98\textwidth}
        \centering\includegraphics[width=\linewidth]{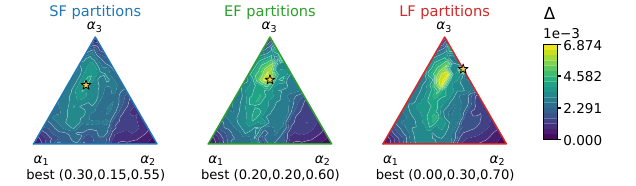}
    \end{subfigure}
    \vspace{0.6em}
    \begin{subfigure}[b]{0.98\textwidth}
        \centering\includegraphics[width=\linewidth]{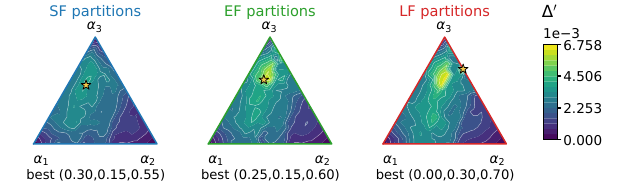}
    \end{subfigure}
    \vspace{0.6em}
    \begin{subfigure}[b]{0.98\textwidth}
        \centering\includegraphics[width=\linewidth]{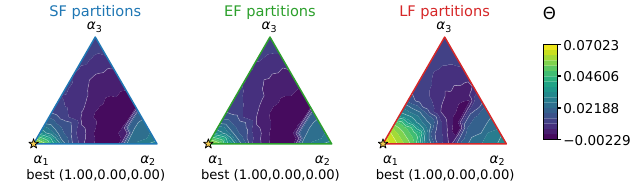}
    \end{subfigure}

    \end{minipage}%
    }
    \caption{Immune Trials multiplex network modularity optimization performed by Leiden-based algorithms. Correction terms and auxiliary quantities. Layer order is $\alpha_1$: intervention, $\alpha_2$: condition, $\alpha_3$: sponsor.}
    \label{fig:ctgov_multiplex_aux}
\end{minipage}
\end{figure}

Tables~\ref{tab:auto_best_results} and~\ref{tab:auto_aux_results} summarize the best observed scores, layer weights, and restart variability. The following results use Leiden throughout; the tables also report Louvain for the two-layer networks. All winners are configurations of source method and layer weights, ranked separately for each target objective.

\begin{table}[p]
\centering\small
\caption{Best median heuristic scores for the main objectives. Each row reports the best median across returned partitions for the stated dataset, optimizer and target objective. Here ``Winner'' means the source method whose returned partitions attain the best median value of the stated target objective, together with the winner location in the simplex and the variability retained at that winning $\alpha$. Dispersion, community counts and runtimes appear in Table~\ref{tab:auto_best_results_dispersion}.}
\label{tab:auto_best_results}
\setlength{\tabcolsep}{3pt}
\begin{tabular}{|lllll|}
\hline
Objective & Winner & Best median & $\alpha$ & Location \\
\hline
\multicolumn{5}{|l|}{\textbf{London Stations --- Leiden}} \\
$Q_{SF}$ & LF & 0.8175$^{\dagger}$ & (1.00, 0.00, 0.00) & vertex \\
$Q_{EF}$ & EF & 0.8282 & (0.40, 0.40, 0.20) & interior \\
$Q_{LF}$ & LF & 0.9379$^{\dagger}$ & (0.70, 0.00, 0.30) & edge \\
\multicolumn{5}{|l|}{\textbf{Political Blogs --- Leiden}} \\
$Q_{SF}$ & SF & 0.5000$^{\dagger}$ & (0.00, 1.00) & vertex \\
$Q_{EF}$ & SF & 0.5000$^{\dagger}$ & (0.00, 1.00) & vertex \\
$Q_{LF}$ & LF & 0.5070$^{\dagger}$ & (1.00, 0.00) & vertex \\
\multicolumn{5}{|l|}{\textbf{Political Blogs --- Louvain}} \\
$Q_{SF}$ & EF & 0.5000$^{\dagger}$ & (0.00, 1.00) & vertex \\
$Q_{EF}$ & SF & 0.5000$^{\dagger}$ & (0.00, 1.00) & vertex \\
$Q_{LF}$ & LF & 0.5060$^{\dagger}$ & (1.00, 0.00) & vertex \\
\multicolumn{5}{|l|}{\textbf{Cora --- Leiden}} \\
$Q_{SF}$ & SF & 0.8220$^{\dagger}$ & (1.00, 0.00) & vertex \\
$Q_{EF}$ & SF & 0.8220$^{\dagger}$ & (1.00, 0.00) & vertex \\
$Q_{LF}$ & LF & 0.9359 & (1.00, 0.00) & vertex \\
\multicolumn{5}{|l|}{\textbf{Cora --- Louvain}} \\
$Q_{SF}$ & EF & 0.8146$^{\dagger}$ & (1.00, 0.00) & vertex \\
$Q_{EF}$ & EF & 0.8146$^{\dagger}$ & (1.00, 0.00) & vertex \\
$Q_{LF}$ & LF & 0.9379 & (1.00, 0.00) & vertex \\
\multicolumn{5}{|l|}{\textbf{Major Airlines --- Leiden}} \\
$Q_{SF}$ & SF & 0.4349$^{\dagger}$ & (0.00, 0.00, 1.00) & vertex \\
$Q_{EF}$ & SF & 0.4349$^{\dagger}$ & (0.00, 0.00, 1.00) & vertex \\
$Q_{LF}$ & LF & 0.8036 & (0.00, 1.00, 0.00) & vertex \\
\multicolumn{5}{|l|}{\textbf{Immune Trials (2L) --- Leiden}} \\
$Q_{SF}$ & LF & 0.8577$^{\dagger}$ & (1.00, 0.00) & vertex \\
$Q_{EF}$ & LF & 0.8577$^{\dagger}$ & (1.00, 0.00) & vertex \\
$Q_{LF}$ & LF & 0.9280 & (1.00, 0.00) & vertex \\
\multicolumn{5}{|l|}{\textbf{Immune Trials (2L) --- Louvain}} \\
$Q_{SF}$ & LF & 0.8567$^{\dagger}$ & (1.00, 0.00) & vertex \\
$Q_{EF}$ & LF & 0.8567$^{\dagger}$ & (1.00, 0.00) & vertex \\
$Q_{LF}$ & LF & 0.9271 & (1.00, 0.00) & vertex \\
\multicolumn{5}{|l|}{\textbf{Immune Trials (3L) --- Leiden}} \\
$Q_{SF}$ & SF & 0.9009$^{\dagger}$ & (0.00, 0.00, 1.00) & vertex \\
$Q_{EF}$ & SF & 0.9009$^{\dagger}$ & (0.00, 0.00, 1.00) & vertex \\
$Q_{LF}$ & LF & 0.9280 & (1.00, 0.00, 0.00) & vertex \\
\hline
\end{tabular}
\end{table}

\begin{table}[p]
\centering\small
\caption{Restart dispersion and runtime for the winning configurations in Table~\ref{tab:auto_best_results}.}
\label{tab:auto_best_results_dispersion}
\setlength{\tabcolsep}{3pt}
\begin{tabular}{|lllll|}
\hline
Objective & IQR & VI@win & Comm. & Time(s) \\
\hline
\multicolumn{5}{|l|}{\textbf{London Stations --- Leiden}} \\
$Q_{SF}$ & 0.0017 & 0.0686 & 117.0 & 0.016 \\
$Q_{EF}$ & 0.0015 & 0.1020 & 21.0 & 0.007 \\
$Q_{LF}$ & 0.0059 & 0.0877 & 83.0 & 0.016 \\
\multicolumn{5}{|l|}{\textbf{Political Blogs --- Leiden}} \\
$Q_{SF}$ & 0.0000 & 0.0000 & 34.0 & 0.247 \\
$Q_{EF}$ & 0.0000 & 0.0000 & 34.0 & 0.247 \\
$Q_{LF}$ & 0.0004 & 0.0067 & 12.0 & 3.050 \\
\multicolumn{5}{|l|}{\textbf{Political Blogs --- Louvain}} \\
$Q_{SF}$ & $4.49\times10^{-6}$ & 0.0007 & 34.0 & 0.201 \\
$Q_{EF}$ & 0.0000 & 0.0003 & 34.0 & 0.369 \\
$Q_{LF}$ & 0.0021 & 0.0170 & 11.5 & 4.531 \\
\multicolumn{5}{|l|}{\textbf{Cora --- Leiden}} \\
$Q_{SF}$ & 0.0008 & 0.0873 & 106.5 & 1.839 \\
$Q_{EF}$ & 0.0008 & 0.0873 & 106.5 & 1.839 \\
$Q_{LF}$ & 0.0026 & 0.0897 & 106.0 & 2.623 \\
\multicolumn{5}{|l|}{\textbf{Cora --- Louvain}} \\
$Q_{SF}$ & 0.0014 & 0.1181 & 104.0 & 0.345 \\
$Q_{EF}$ & 0.0014 & 0.1181 & 104.0 & 0.345 \\
$Q_{LF}$ & 0.0064 & 0.1118 & 104.5 & 5.300 \\
\multicolumn{5}{|l|}{\textbf{Major Airlines --- Leiden}} \\
$Q_{SF}$ & 0.0000 & 0.0000 & 67.0 & 0.003 \\
$Q_{EF}$ & 0.0000 & 0.0000 & 67.0 & 0.003 \\
$Q_{LF}$ & 0.0206 & 0.0934 & 28.0 & 0.005 \\
\multicolumn{5}{|l|}{\textbf{Immune Trials (2L) --- Leiden}} \\
$Q_{SF}$ & 0.0001 & 0.0107 & 136.0 & 0.301 \\
$Q_{EF}$ & 0.0001 & 0.0107 & 136.0 & 0.301 \\
$Q_{LF}$ & 0.0007 & 0.0107 & 136.0 & 0.301 \\
\multicolumn{5}{|l|}{\textbf{Immune Trials (2L) --- Louvain}} \\
$Q_{SF}$ & 0.0010 & 0.0203 & 135.0 & 0.422 \\
$Q_{EF}$ & 0.0010 & 0.0203 & 135.0 & 0.422 \\
$Q_{LF}$ & 0.0016 & 0.0203 & 135.0 & 0.422 \\
\multicolumn{5}{|l|}{\textbf{Immune Trials (3L) --- Leiden}} \\
$Q_{SF}$ & $1.11\times10^{-16}$ & 0.0002 & 238.0 & 0.208 \\
$Q_{EF}$ & 0.0000 & 0.0002 & 238.0 & 0.208 \\
$Q_{LF}$ & 0.0007 & 0.0107 & 136.0 & 0.270 \\
\hline
\end{tabular}
\par\smallskip\raggedright\footnotesize IQR is the interquartile width. Time(s) is the median per-configuration runtime; for LF, it includes the layer-wise phase-one cost amortized uniformly over all grid points. Comm. denotes the median number of communities at the winning grid point; VI@win denotes mean pairwise normalized variation of information across restarts. A dagger ($^{\dagger}$) on a Best-median value marks a winner--runner-up median-difference interval that includes zero, including $[0,0]$ for tied constant scores. This bootstrap interval summarizes variation across seed restarts; it does not establish statistical equivalence of methods. The two configurations may use the same $\alpha$.

\end{table}

\begin{table}[p]
\centering\small
\caption{Best median heuristic scores for auxiliary correction terms. These rows summarize where the EF and LF correction terms are largest and should not be read as main objective winners. Dispersion, community counts and runtimes appear in Table~\ref{tab:auto_aux_results_dispersion}.}
\label{tab:auto_aux_results}
\setlength{\tabcolsep}{3pt}
\begin{tabular}{|lllll|}
\hline
Objective & Winner & Best median & $\alpha$ & Location \\
\hline
\multicolumn{5}{|l|}{\textbf{London Stations --- Leiden}} \\
$\Delta$ & EF & 0.1160$^{\dagger}$ & (0.15, 0.45, 0.40) & interior \\
$\Delta'$ & EF & 0.1145$^{\dagger}$ & (0.15, 0.45, 0.40) & interior \\
$\Theta$ & LF & 0.1522$^{\dagger}$ & (0.00, 0.00, 1.00) & vertex \\
\multicolumn{5}{|l|}{\textbf{Political Blogs --- Leiden}} \\
$\Delta$ & LF & $1.63\times10^{-5}$ & (0.50, 0.50) & interior \\
$\Delta'$ & LF & $1.56\times10^{-5}$ & (0.55, 0.45) & interior \\
$\Theta$ & LF & 0.0799$^{\dagger}$ & (1.00, 0.00) & vertex \\
\multicolumn{5}{|l|}{\textbf{Political Blogs --- Louvain}} \\
$\Delta$ & LF & $1.69\times10^{-5}$$^{\dagger}$ & (0.50, 0.50) & interior \\
$\Delta'$ & LF & $1.64\times10^{-5}$$^{\dagger}$ & (0.55, 0.45) & interior \\
$\Theta$ & LF & 0.0792 & (1.00, 0.00) & vertex \\
\multicolumn{5}{|l|}{\textbf{Cora --- Leiden}} \\
$\Delta$ & EF & 0.0013 & (0.45, 0.55) & interior \\
$\Delta'$ & EF & 0.0013$^{\dagger}$ & (0.50, 0.50) & interior \\
$\Theta$ & LF & 0.1781 & (0.00, 1.00) & vertex \\
\multicolumn{5}{|l|}{\textbf{Cora --- Louvain}} \\
$\Delta$ & EF & 0.0011$^{\dagger}$ & (0.40, 0.60) & interior \\
$\Delta'$ & EF & 0.0012$^{\dagger}$ & (0.45, 0.55) & interior \\
$\Theta$ & LF & 0.1803 & (0.00, 1.00) & vertex \\
\multicolumn{5}{|l|}{\textbf{Major Airlines --- Leiden}} \\
$\Delta$ & EF & 0.1472$^{\dagger}$ & (0.55, 0.35, 0.10) & interior \\
$\Delta'$ & EF & 0.1068$^{\dagger}$ & (0.55, 0.35, 0.10) & interior \\
$\Theta$ & LF & 0.5574 & (0.00, 1.00, 0.00) & vertex \\
\multicolumn{5}{|l|}{\textbf{Immune Trials (2L) --- Leiden}} \\
$\Delta$ & EF & 0.0030$^{\dagger}$ & (0.70, 0.30) & interior \\
$\Delta'$ & LF & 0.0034$^{\dagger}$ & (0.40, 0.60) & interior \\
$\Theta$ & LF & 0.1322 & (0.00, 1.00) & vertex \\
\multicolumn{5}{|l|}{\textbf{Immune Trials (2L) --- Louvain}} \\
$\Delta$ & LF & 0.0028 & (0.35, 0.65) & interior \\
$\Delta'$ & LF & 0.0031 & (0.35, 0.65) & interior \\
$\Theta$ & LF & 0.1205$^{\dagger}$ & (0.00, 1.00) & vertex \\
\multicolumn{5}{|l|}{\textbf{Immune Trials (3L) --- Leiden}} \\
$\Delta$ & LF & 0.0069$^{\dagger}$ & (0.00, 0.30, 0.70) & edge \\
$\Delta'$ & LF & 0.0068$^{\dagger}$ & (0.00, 0.30, 0.70) & edge \\
$\Theta$ & LF & 0.0702 & (1.00, 0.00, 0.00) & vertex \\
\hline
\end{tabular}
\end{table}

\begin{table}[p]
\centering\small
\caption{Restart dispersion and runtime for the winning configurations in Table~\ref{tab:auto_aux_results}.}
\label{tab:auto_aux_results_dispersion}
\setlength{\tabcolsep}{3pt}
\begin{tabular}{|lllll|}
\hline
Objective & IQR & VI@win & Comm. & Time(s) \\
\hline
\multicolumn{5}{|l|}{\textbf{London Stations --- Leiden}} \\
$\Delta$ & 0.0114 & 0.0722 & 19.0 & 0.007 \\
$\Delta'$ & 0.0109 & 0.0722 & 19.0 & 0.007 \\
$\Theta$ & 0.0217 & 0.0143 & 331.0 & 0.007 \\
\multicolumn{5}{|l|}{\textbf{Political Blogs --- Leiden}} \\
$\Delta$ & $5.17\times10^{-7}$ & 0.0067 & 12.0 & 3.432 \\
$\Delta'$ & $8.51\times10^{-8}$ & 0.0067 & 12.0 & 3.748 \\
$\Theta$ & 0.0003 & 0.0067 & 12.0 & 3.050 \\
\multicolumn{5}{|l|}{\textbf{Political Blogs --- Louvain}} \\
$\Delta$ & $2.69\times10^{-6}$ & 0.0164 & 12.0 & 5.250 \\
$\Delta'$ & $3.09\times10^{-6}$ & 0.0168 & 11.0 & 5.516 \\
$\Theta$ & 0.0015 & 0.0170 & 11.5 & 4.531 \\
\multicolumn{5}{|l|}{\textbf{Cora --- Leiden}} \\
$\Delta$ & 0.0001 & 0.0285 & 10.0 & 0.685 \\
$\Delta'$ & $8.28\times10^{-5}$ & 0.0365 & 10.0 & 0.765 \\
$\Theta$ & 0.0011 & 0.0045 & 7.0 & 6.407 \\
\multicolumn{5}{|l|}{\textbf{Cora --- Louvain}} \\
$\Delta$ & 0.0001 & 0.0514 & 10.0 & 2.278 \\
$\Delta'$ & 0.0002 & 0.0614 & 9.0 & 2.874 \\
$\Theta$ & 0.0008 & 0.0167 & 7.0 & 14.956 \\
\multicolumn{5}{|l|}{\textbf{Major Airlines --- Leiden}} \\
$\Delta$ & 0.0146 & 0.0441 & 5.0 & 0.003 \\
$\Delta'$ & 0.0059 & 0.0441 & 5.0 & 0.003 \\
$\Theta$ & 0.0205 & 0.0934 & 28.0 & 0.005 \\
\multicolumn{5}{|l|}{\textbf{Immune Trials (2L) --- Leiden}} \\
$\Delta$ & 0.0003 & 0.0374 & 15.0 & 0.199 \\
$\Delta'$ & 0.0004 & 0.0803 & 10.0 & 1.329 \\
$\Theta$ & 0.0280 & 0.0416 & 6.0 & 1.069 \\
\multicolumn{5}{|l|}{\textbf{Immune Trials (2L) --- Louvain}} \\
$\Delta$ & 0.0005 & 0.0950 & 9.0 & 2.594 \\
$\Delta'$ & 0.0006 & 0.0950 & 9.0 & 2.594 \\
$\Theta$ & 0.0302 & 0.0511 & 6.0 & 2.250 \\
\multicolumn{5}{|l|}{\textbf{Immune Trials (3L) --- Leiden}} \\
$\Delta$ & 0.0055 & 0.0609 & 15.0 & 1.078 \\
$\Delta'$ & 0.0054 & 0.0609 & 15.0 & 1.078 \\
$\Theta$ & 0.0008 & 0.0107 & 136.0 & 0.270 \\
\hline
\end{tabular}
\par\smallskip\raggedright\footnotesize IQR is the interquartile width. Time(s) is the median per-configuration runtime; for LF, it includes the layer-wise phase-one cost amortized uniformly over all grid points. Comm. denotes the median number of communities at the winning grid point; VI@win denotes mean pairwise normalized variation of information across restarts. A dagger ($^{\dagger}$) on a Best-median value marks a winner--runner-up median-difference interval that includes zero, including $[0,0]$ for tied constant scores. This bootstrap interval summarizes variation across seed restarts; it does not establish statistical equivalence of methods. The two configurations may use the same $\alpha$.

\end{table}

\paragraph{London Stations.} The best median $Q_{\sf SF}$ is 0.8175 at $\alpha=\allowbreak(1.00, 0.00, 0.00)$ (LF partitions). The best median $Q_{\sf EF}$ is 0.8282 at $\alpha=\allowbreak(0.40, 0.40, 0.20)$ (EF partitions; interior). Its exploratory bootstrap interval is [0.8276, 0.8284]; the gap to the runner-up configuration (EF, $\alpha=\allowbreak(0.40, 0.35, 0.25)$) has interval [0.0003, 0.0013]. The largest median $\Delta$ is 0.1160 at $\alpha=\allowbreak(0.15, 0.45, 0.40)$ (interior); $\Delta'$ reaches 0.1145 at $\alpha=\allowbreak(0.15, 0.45, 0.40)$ (interior). On the transient-graph scale, $Q_{\sf LF}$ reaches 0.9379 at $\alpha=\allowbreak(0.70, 0.00, 0.30)$ (LF partitions).

\paragraph{Political Blogs.} The best median $Q_{\sf SF}$ is 0.5000 at $\alpha=\allowbreak(0.00, 1.00)$ (SF partitions). The best median $Q_{\sf EF}$ is 0.5000 at $\alpha=\allowbreak(0.00, 1.00)$ (SF partitions; vertex). Its exploratory bootstrap interval is [0.5000, 0.5000]; the gap to the runner-up configuration (EF, $\alpha=\allowbreak(0.00, 1.00)$) has interval [0.0000, 0.0000]. The largest median $\Delta$ is $1.63\times10^{-5}$ at $\alpha=\allowbreak(0.50, 0.50)$ (interior); $\Delta'$ reaches $1.56\times10^{-5}$ at $\alpha=\allowbreak(0.55, 0.45)$ (interior). On the transient-graph scale, $Q_{\sf LF}$ reaches 0.5070 at $\alpha=\allowbreak(1.00, 0.00)$ (LF partitions).

\paragraph{Cora.} The best median $Q_{\sf SF}$ is 0.8220 at $\alpha=\allowbreak(1.00, 0.00)$ (SF partitions). The best median $Q_{\sf EF}$ is 0.8220 at $\alpha=\allowbreak(1.00, 0.00)$ (SF partitions; vertex). Its exploratory bootstrap interval is [0.8218, 0.8223]; the gap to the runner-up configuration (LF, $\alpha=\allowbreak(1.00, 0.00)$) has interval [-0.0004, 0.0004]. The largest median $\Delta$ is 0.0013 at $\alpha=\allowbreak(0.45, 0.55)$ (interior); $\Delta'$ reaches 0.0013 at $\alpha=\allowbreak(0.50, 0.50)$ (interior). On the transient-graph scale, $Q_{\sf LF}$ reaches 0.9359 at $\alpha=\allowbreak(1.00, 0.00)$ (LF partitions).

\paragraph{Major Airlines.} The best median $Q_{\sf SF}$ is 0.4349 at $\alpha=\allowbreak(0,0,1)$ (SF partitions). The best median $Q_{\sf EF}$ is 0.4349 at $\alpha=\allowbreak(0,0,1)$ (SF partitions; vertex). Its exploratory bootstrap interval is [0.4349, 0.4349]; the gap to the runner-up configuration (EF, $\alpha=\allowbreak(0,0,1)$) has interval [0.0000, 0.0000]. The largest median $\Delta$ is 0.1472 at $\alpha=\allowbreak(0.55, 0.35, 0.10)$ (interior); $\Delta'$ reaches 0.1068 at $\alpha=\allowbreak(0.55, 0.35, 0.10)$ (interior). On the transient-graph scale, $Q_{\sf LF}$ reaches 0.8036 at $\alpha=\allowbreak(0.00, 1.00, 0.00)$ (LF partitions).

\paragraph{Immune Trials (2L).} The best median $Q_{\sf SF}$ is 0.8577 at $\alpha=\allowbreak(1.00, 0.00)$ (LF partitions). The best median $Q_{\sf EF}$ is 0.8577 at $\alpha=\allowbreak(1.00, 0.00)$ (LF partitions; vertex). Its exploratory bootstrap interval is [0.8577, 0.8578]; the gap to the runner-up configuration (SF, $\alpha=\allowbreak(1.00, 0.00)$) has interval [$-3.35\times10^{-5}$, $8.13\times10^{-5}$]. The largest median $\Delta$ is 0.0030 at $\alpha=\allowbreak(0.70, 0.30)$ (interior); $\Delta'$ reaches 0.0034 at $\alpha=\allowbreak(0.40, 0.60)$ (interior). On the transient-graph scale, $Q_{\sf LF}$ reaches 0.9280 at $\alpha=\allowbreak(1.00, 0.00)$ (LF partitions).

\paragraph{Immune Trials (3L).} The best median $Q_{\sf SF}$ is 0.9009 at $\alpha=\allowbreak(0,0,1)$ (SF partitions). The best median $Q_{\sf EF}$ is 0.9009 at $\alpha=\allowbreak(0,0,1)$ (SF partitions; vertex). Its exploratory bootstrap interval is [0.9009, 0.9009]; the gap to the runner-up configuration (EF, $\alpha=\allowbreak(0,0,1)$) has interval [0.0000, 0.0000]. The largest median $\Delta$ is 0.0069 at $\alpha=\allowbreak(0.00, 0.30, 0.70)$ (edge); $\Delta'$ reaches 0.0068 at $\alpha=\allowbreak(0.00, 0.30, 0.70)$ (edge). On the transient-graph scale, $Q_{\sf LF}$ reaches 0.9280 at $\alpha=\allowbreak(1.00, 0.00, 0.00)$ (LF partitions).

\section{Label-free Cora: numerical detail}
\label{app:cora_detail}
\begin{table}[H]
\centering\small
\setlength{\tabcolsep}{2pt}
\caption{Label-free Cora baselines: words only ($a=0$), equal weights ($a=0.5$), and citations only ($a=1$). Entries are medians [first quartile, third quartile] over 30 seeds. All objective-selected configurations have $a=1$. LF endpoints use transformed layers; EF and SF endpoints use original layers.}
\label{tab:cora_external_baselines}
\label{tab:cora_external_selected}
\begin{tabular}{|llrll|}
\hline
Optimizer & Method & $a$ & ARI & NMI \\
\hline
Louvain & SF & 0.00 & 0.301 [0.293, 0.310] & 0.374 [0.368, 0.380] \\
Louvain & SF & 0.50 & 0.427 [0.412, 0.442] & 0.546 [0.541, 0.553] \\
Louvain & SF & 1.00 & 0.249 [0.228, 0.258] & 0.450 [0.444, 0.453] \\
Louvain & EF & 0.00 & 0.302 [0.296, 0.309] & 0.373 [0.366, 0.379] \\
Louvain & EF & 0.50 & 0.431 [0.422, 0.450] & 0.552 [0.544, 0.559] \\
Louvain & EF & 1.00 & 0.250 [0.230, 0.262] & 0.455 [0.452, 0.458] \\
Louvain & LF & 0.00 & 0.304 [0.294, 0.307] & 0.375 [0.368, 0.381] \\
Louvain & LF & 0.50 & 0.404 [0.386, 0.422] & 0.499 [0.483, 0.511] \\
Louvain & LF & 1.00 & 0.228 [0.223, 0.257] & 0.450 [0.448, 0.456] \\
Leiden & SF & 0.00 & 0.311 [0.307, 0.315] & 0.380 [0.374, 0.384] \\
Leiden & SF & 0.50 & 0.448 [0.442, 0.453] & 0.567 [0.562, 0.571] \\
Leiden & SF & 1.00 & 0.256 [0.249, 0.262] & 0.460 [0.457, 0.466] \\
Leiden & EF & 0.00 & 0.310 [0.307, 0.316] & 0.379 [0.376, 0.383] \\
Leiden & EF & 0.50 & 0.444 [0.440, 0.450] & 0.564 [0.559, 0.569] \\
Leiden & EF & 1.00 & 0.256 [0.247, 0.260] & 0.462 [0.457, 0.465] \\
Leiden & LF & 0.00 & 0.311 [0.304, 0.314] & 0.381 [0.377, 0.384] \\
Leiden & LF & 0.50 & 0.381 [0.375, 0.391] & 0.465 [0.459, 0.470] \\
Leiden & LF & 1.00 & 0.259 [0.249, 0.262] & 0.461 [0.457, 0.465] \\
\hline
\end{tabular}
\end{table}
\noindent\begin{minipage}{\linewidth}\paragraph{Community counts.} At the objective-selected weights, the median [first quartile, third quartile] community counts are Louvain SF: 104.0 [103.0, 105.0]; Louvain EF: 104.0 [103.0, 105.0]; Louvain LF: 104.5 [103.2, 105.8]; Leiden SF: 106.0 [105.0, 107.0]; Leiden EF: 106.5 [105.0, 107.0]; Leiden LF: 106.0 [106.0, 107.0]. The reference annotations contain seven categories. Counts for every weight and seed are included in the electronic supplement.
\end{minipage}

\section{Cora Sensitivity Audit}
\label{app:cora_sensitivity}

\paragraph{Historical Cora sensitivity.} This audit uses the historical attribute representation, including subject-category indicators; it does not evaluate external label agreement. Table~\ref{tab:cora_sensitivity} audits whether the EF winner location is sensitive to the cosine k-NN parameter and symmetrization rule used for the Cora attribute layer. Its purpose is narrowly to test whether the best-median EF configuration remains at the same simplex location under these four preprocessing variants. Unlike the main runs (30 seeds with seed-resampling intervals), this audit uses 10 seeds per grid point and reports point estimates (medians and restart dispersion) only, without bootstrap intervals; it is a qualitative robustness check rather than a statistically powered comparison. Across all four variants the best median $Q_{EF}$ stays at the pure-structural vertex $\alpha=(1,0)$. At that vertex, the largest within-variant spread among the $Q_{EF}$ medians attained by partitions returned by SF, EF, and LF is 0.0008 when rounded to four decimal places; these 10-seed point estimates do not support a ranking of source methods. Only the k=25 union variant uses the canonical processed attribute layer; the other three variants were constructed for this sensitivity audit.
\begin{table}[htbp]
    \centering
    \small\setlength{\tabcolsep}{3pt}
    \caption{Cora sensitivity audit under Leiden. Winners are reported for the $Q_{EF}$ target objective only, because the appendix is intended to stress-test the preprocessing choice that most directly affects the EF geometry. Every winner is at the vertex $\alpha=(1,0)$.}
    \label{tab:cora_sensitivity}
    \begin{tabular}{|lllll|}
        \hline
        Variant & Edges & Winner & Best median $Q_{EF}$ & VI@win \\
        \hline
        $k=10$, union & 20588 & SF & 0.8224 & 0.0870 \\
        $k=25$, union & 49314 & LF & 0.8220 & 0.0876 \\
        $k=50$, union & 95379 & SF & 0.8225 & 0.0861 \\
        $k=25$, mutual & 18352 & SF & 0.8223 & 0.0869 \\
        \hline
    \end{tabular}
\end{table}

\section{A normalized cluster-network alternative inspired by \texorpdfstring{\cite{Liu2020}}{Liu et al.}}
\label{app:liu_lf}

Liu et al.\ \cite{Liu2020} obtain partitions from network structure and node attributes separately, encode them by binary co-membership matrices, and fuse the matrices linearly before applying community detection. For numerical attributes, their attribute partition can come from a clustering method rather than modularity maximization. We compare this idea through the following normalized, loop-free, $L$-layer adaptation within our modularity framework (Figure~\ref{fig:scheme_late}); it is not an exact restatement of their algorithm.

See e.g. \cite{Shamir2004,Pandove2017} for other problems involving cluster networks.

Starting from layer-wise partitions $C'_l$, replace each community by a clique to obtain a cluster network $g'_l$. Let $S_l=\sum_k\binom{|C'_{l,k}|}{2}$ be the number of ordinary within-community edges. For $S_l>0$, normalize their indicator weights to total edge weight one:

\begin{equation*}
    \label{new_weights}
    w'_l(e_{ij}) =
    \begin{cases}
        \frac{1}{S_l}, \quad \text{if $i\ne j$ and $v_i, v_j \in C'_{l,k}$ for some }k,\\
        0, \quad \text{otherwise},
    \end{cases}
\end{equation*}
The all-singleton case has $S_l=0$ and is excluded from this loop-free normalized adaptation. The per-layer normalization changes the numerical interpretation of the fusion weights relative to the original co-membership matrices.

The final step applies EF to the cluster multiplex $G'$ with the chosen $\alpha$.

Our absorbing construction and this normalized cluster-network adaptation assign different edge weights. The latter need not satisfy (\ref{LF_construction_1}), as Figure~\ref{fig:example-LF} shows, so the monotonicity argument in Lemma~\ref{lemma_lf_connected_to_sf} does not carry over directly.

\begin{figure}[H]
\centering\begin{minipage}{4.25in}
    \centering
    \begin{minipage}[b]{\linewidth}
        \centering
        \includegraphics[width=0.765\textwidth]{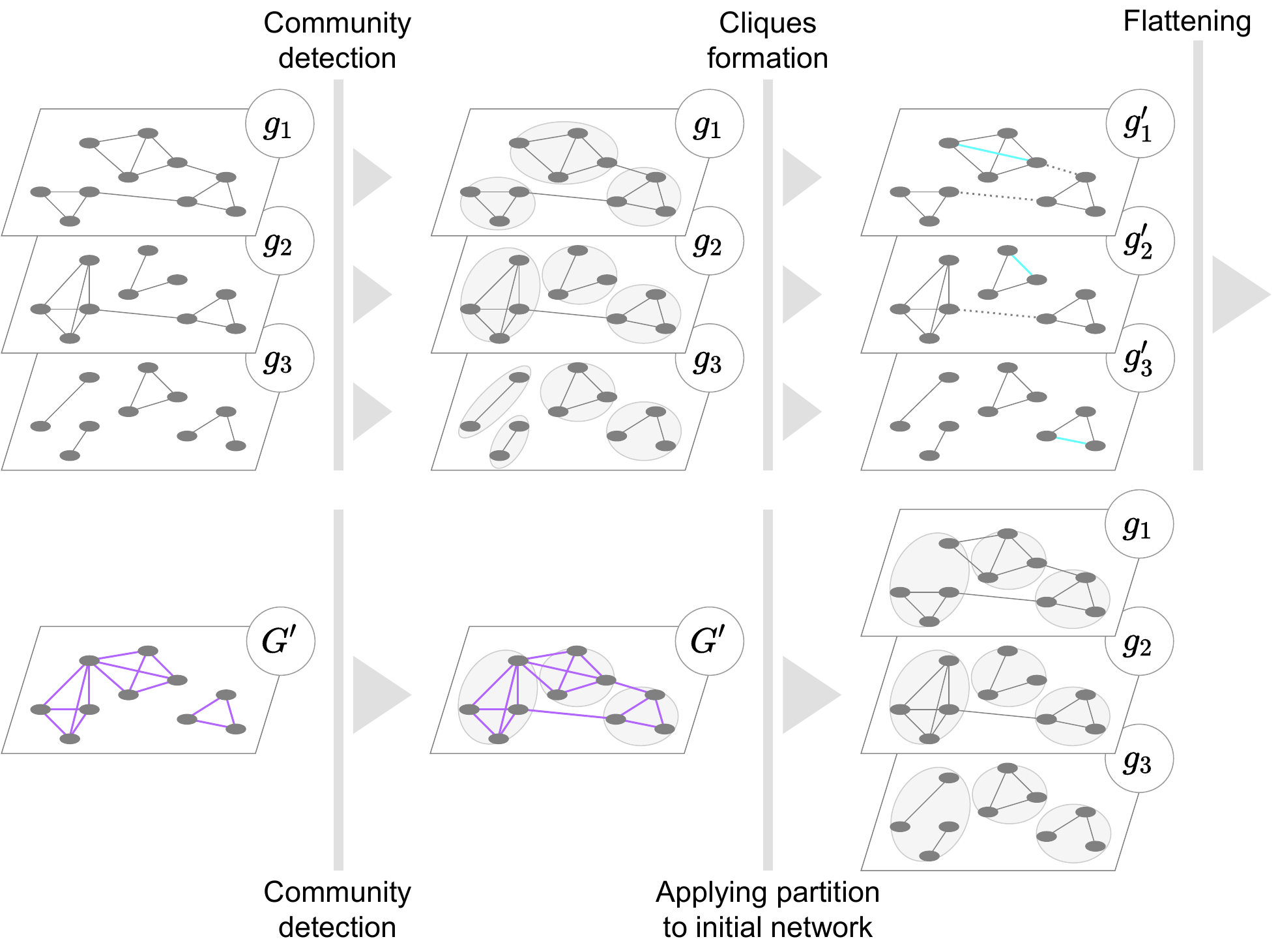}
        \\
        \includegraphics[width=0.05\linewidth]{figures/grey_line.pdf}
        link
        \includegraphics[width=0.05\linewidth]{figures/sky_blue_line.pdf}
        added link
        \includegraphics[width=0.05\linewidth]{figures/blue_line.pdf}
        updated link
        \includegraphics[width=0.05\linewidth]{figures/purple_line.pdf}
        fused link
        \includegraphics[width=0.05\linewidth]{figures/grey_dot_line.pdf}
        erased link
    \end{minipage}
    \caption{The normalized cluster-network adaptation of the LF idea in \cite{Liu2020}.}
    \label{fig:scheme_late}
\end{minipage}
\end{figure}

\begin{figure}[H]
\centering\begin{minipage}{4.25in}
    \centering
    \includegraphics[width=0.34\linewidth]{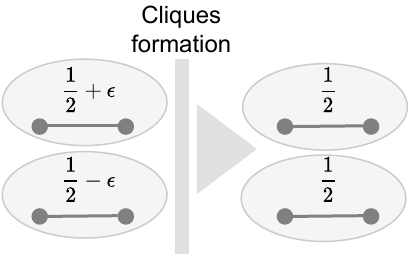}
    \caption{Example for the normalized cluster-network adaptation, where condition (\ref{LF_construction_1}) is not satisfied, $\varepsilon \in (0,\tfrac{1}{2})$.}
    \label{fig:example-LF}
\end{minipage}
\end{figure}

\end{appendices}

\end{document}